\documentclass[preprint,10pt]{elsarticle}

\usepackage[fleqn]{amsmath}
\usepackage{algorithm}
\usepackage{orcidlink}
\usepackage{algpseudocode}
\usepackage{amssymb}
\usepackage{amsthm}
\usepackage{graphicx}
\usepackage[outdir=./]{epstopdf}
\usepackage{caption}
\usepackage{pdfpages}
\usepackage{adjustbox}
\usepackage{rotating}
\usepackage[T1]{fontenc}
\usepackage{amssymb}
\usepackage{mathrsfs}
\usepackage{multirow}
\usepackage{hhline}
\usepackage{enumerate}
\usepackage{algpseudocode}
\usepackage{etoolbox}
\usepackage{enumerate}
\usepackage{enumitem}
\usepackage{makecell}
\usepackage{tabto}
\usepackage{pdflscape}

\usepackage{algpseudocode}
\floatname{algorithm}{Procedure}

\algnewcommand\algorithmicset{\textbf{SET}}
\algnewcommand\algorithmicto{\textbf{TO}}
\algnewcommand\SET[2]{\State\algorithmicset\ #1 \algorithmicto\ #2}
\algnewcommand\algorithmicreceive{\textbf{RECEIVE}}
\algnewcommand\algorithmicfromkeyboard{\textbf{FROM}}
\algnewcommand\RECEIVE[1]{\State\algorithmicreceive\ #1 \algorithmicfromkeyboard}
\algnewcommand\algorithmicsend{\textbf{send}}
\algnewcommand\algorithmictodisplay{\textbf{}}
\algnewcommand\SEND[1]{\State\algorithmicsend\ #1 \algorithmictodisplay}

\usepackage{algorithm}
\journal{ }

\newtheorem{proposition}{Proposition}

\begin{document}
\begin{frontmatter}
		\title{In-Memory Sorting-Searching with Cayley Tree\tnoteref{t1}}
		\tnotetext[t1]{This work is supported by CRG project (File No.: CRG$/2023/006799$) of SERB, Govt. of India.}
		
		\author[1]{Subrata Paul\corref{cor1}~\orcidlink{0000-0001-8919-7193}}
		\ead{psubrata.it@gmail.com}
		
		\author[1]{Sukanta Das~\orcidlink{0000-0001-6110-5082}}
		\ead{sukanta@it.iiests.ac.in}
		
		\author[2]{Biplab~K~Sikdar~\orcidlink{0000-0002-9394-8540}}
		\ead{biplab@cs.iiests.ac.in}
		
		\cortext[cor1]{Corresponding author}
		\affiliation[1]{organization={Department of Information Technology},
			addressline={Indian Institute of Enginnering Science and Technology}, 
			city={Shibpur},
			country={India}}
		\affiliation[2]{organization={Department of Computer Science and Technology},
			addressline={Indian Institute of Enginnering Science and Technology}, 
			city={Shibpur},
			country={India}}
			
	\begin{abstract}
		This work proposes a computing model to reduce the workload of CPU. It relies on the data intensive computation in memory, where the data reside, and effectively realizes an in-memory computing (IMC) platform. Each memory word, with additional logic, acts as a tiny processing element which forms the node of a Cayley tree. The Cayley tree in turn defines the framework for solving the data intensive computational problems. It finds the solutions for {\em in-memory searching}, computing the {\em max (min) in-memory} and {\em in-memory sorting} while reducing the involvement of CPU. The worst case time complexities of the IMC based solutions for {\em in-memory searching} and {\em computing max (min) in-memory} are $\mathcal{O}\log{n}$. Such solutions are independent of the order of elements in the list. The worst case time complexity of {\em in-memory sorting}, on the other hand, is $\mathcal{O}(n\log{n})$. Two types of hardware implementations of the IMC platform are proposed. One is based on the existing/conventional memory architecture, and the other one is on a newly defined memory architecture. The solutions are further implemented in FPGA platform to prove the effectiveness of the IMC architecture while comparing with the state-of-the art designs.
	\end{abstract}
	\begin{keyword}
		In-memory computing, Cayley tree, In-memory searching, Computing max/min in-memory, In-memory sorting, Data intensive computation, Cellular Automata
	\end{keyword}
\end{frontmatter}	

	\section{Introduction}
	
	The in-memory computing (IMC) architecture empowers memory array to perform arithmetic, logical and other operations formally done by the CPU~\cite{sebastian2020memory}. In conventional von Neumann architecture, the data transfer from memory to CPU and vice versa is frequent~\cite{hashemi2016continuous}. On the other hand, an IMC architecture avoids such data transfer, resulting speedy computation~\cite{mutlu2019processing}. This has sparked significant research interests among the researchers, working in the field of computer architecture, to explore designs that can carry out computations inside the memory cells~\cite{yang2022novel}. The memory of an ideal IMC system preferably be the non-volatile, multi-bit, with extended cycle durability, and ease of read/write operations~\cite{sebastian2020memory,wong2015memory}. A number of IMC platforms around the static RAM (SRAM) and dynamic RAM (DRAM) have also been proposed in the literature~\cite{agrawal2018,ali2019}. 
	
	The State-of-the-art CMOS technology features unmanageable device leakage current. Hence, the power wall and memory wall bottleneck in modern computing platform may not be fully resolved~\cite{sandhie2021investigation}. One of the most promising technologies is the spintronic-based magnetic random access memory (MRAM) which offers advantages such as non-volatility, data retention, durability during reprogramming, etc~\cite{li2022survey}. As the potential IMC platform, nonvolatile memories (NVMs) have drawn a lot of interest~\cite{huang2016reconfigurable,jain2017computing,kang2017memory,le2018compressed,zhang2019spintronic,shreya2020computing,wang2020design}. The MRAM can be a choice for IMC platform~\cite{li2022survey}. With the NVMs, the energy-efficient computing may be achieved by adopting some changes in the peripheral circuitry. This effectively leads to huge leakage power reduction. 
	
	The IMC, however, is restricted to a few basic operations as only a few logic functions can be implemented in the current IMC systems. The main difficulty with the IMC is that it cannot be implemented without compromising the standard memory performance, such as the read/write speed and storage density. In this scenario, this work proposes an IMC architecture to satisfy the need of data intensive computation. Effectiveness of the architecture is verified while realizing the {\em in-memory sorting}, {\em in-memory searching} and {\em computing max/min in-memory}. The words in a memory module are organized to form the nodes of a Cayley tree and that realizes the IMC.
	
	In general, a Cayley tree is an infinite tree, each node of which has a fixed degree. However, this work considers a finite Cayley tree, where all but terminating nodes have the same number of connected neighbors~\cite{cayley1878desiderata} and all the leaves maintain the same distance from the {\em root}. The Cayley tree formed with the memory word, with a set of added flags, is introduced in this paper to realize the three schemes -- {\em in-memory searching}, {\em computing max (min) in-memory} and {\em in-memory sorting}. The proposed IMC platform is presented in Section~\ref{the-model}, where the Cayley tree structure are defined following a brief on related works in Section~\ref{related-work}. The {\em in-memory searching} scheme, developed around the Cayley tree structure, is detailed out in Section~\ref{case1}. The {\em computing max (min) in-memory} is introduced in Section~\ref{sec-max}. Section~\ref{sec-sorting} elaborates the {\em in-memory sorting} scheme. The field-programmable gate array (FPGA) implementation of the IMC platforms are reported in Section~\ref{hardware}. Section~\ref{conclusion} concludes the work.
	
	\section{Related Works}\label{related-work}
	 The Computing-In-Memory (CIM) paradigm that emerged in the post-Moore era is a deviation from the conventional von Neumann architecture~\cite{godfrey1993introduction}. The CIM performs part of the operations within or near to the memory arrays that conventionally executed in the CPU. It reduces the expenditure of data movement between the CPU and memory~\cite{asifuzzaman2023survey}. 
	 
	 Depending on the proximity of computation, near to memory as well as the responsibility of memory in computation, the CIM is categorized as the near-memory processing (NMP), processing-in-memory (PIM) and in-memory computing (IMC)~\cite{bao2022toward,sun2023survey}. The PIM architecture has been proposed in~\cite{wang2015general,nai2017graphpim,guo2021spintronics} as a solution to the von Neumann bottleneck of huge data movement between the CPU and memory. It fits well with the newly developed non-volatile spin-transfer torque magnetic RAM (STT MRAM) that has several benefits including the fast read/write and low energy consumption~\cite{jain2017computing,kang2017memory}. 
	 A proposal for triangle counting acceleration with processing-in-mram (TCIM) can be found in~\cite{wang2020tcim}. The TCIM performs bit-wise logical operations.
	 
	 In IMC, the operations that are conventionally executed within the CPU are executed in memory arrays. A number of designs have been proposed around the volatile and non-volatile memory, realizing the IMC~\cite{jain2017computing,sun2023survey}. However, the most popular designs are based on the memristor-based material implication (IMPLY)~\cite{borghetti2010memristive,kvatinsky2013memristor} and Memristor aided logic (MAGIC)~\cite{kvatinsky2014magic}. The IMPLY has been demonstrated to be useful for logic synthesis in~\cite{borghetti2010memristive,shirinzadeh2017logic}. The memristive IMC can be realized in the memristive crossbar array~\cite{batcher1968sorting,peters2011fast,peters2012novel}. The work of~\cite{batcher1968sorting} reports realization of bitonic sorting in memristive crossbar arrays. Several techniques for in-memory graph processing have been also proposed in the IMC platform~\cite{shim2022gp3d,chen2022accelerating,wei2023imga}. However, the IMC around the memristor is yet to be introduced as the main stream technology, replacing the CMOS.

	\vspace{1 em}

	\section{The In-memory Computing Platform}\label{the-model}
	The in-memory computing (IMC) platform, proposed in this work, is based on Cayley tree. Classically, a Cayley tree of order $\eta\geq 1$ is an infinite tree in which each non-leaf node has a constant number of branches. 
	
	In this work, we introduce finite Cayley tree as an IMC platform. All the leaf nodes of the tree maintain the same distance from the root. The distance of a leaf node from the {\em root} is called the height ($h$) of the tree. Each non-leaf (except the root) node has $\eta$ children. Hence, all such nodes have same degree, i.e., $\eta+1$. For example, a Cayley tree of height $3$ is shown in Figure~\ref{tree-cell}(a), where the order of the tree $\eta$ is $2$. The red (lightly shaded) node is the {\em root} and the nodes at its next level (say level $1$) are the children of the root. The intermediate nodes --that is, other than the {\em root} and leaves, have two children (left child and right child) --that is, the degree of each node is $3$.
	\vspace{1 em}
	\begin{proposition}\label{th1}
		In a (finite) Cayley tree of order $\eta$ and height $h\geq 1$, the number of nodes ($n$) is	
			\begin{equation}\label{eq1}
			n = \begin{cases}
				1  &\text{if } h = 1 \\
				1 + (\eta+1)\sum_{i=0}^{h-2} \eta^{i}&\text{otherwise}
			\end{cases}\\
		\end{equation}
	\end{proposition}
	\begin{proof}
		Let consider a Cayley tree of order $\eta$. That is, the {\em root} has $(\eta+1)$ branches and an intermediate node has $\eta$ children and one parent. A leaf, on the other hand, is having only the parent. Now, the growth of the number of nodes $n$ in the Cayley tree with height $h$ is as following
		\begin{align*}
			\text{Height}&&\text{Number of nodes ($n$)}\\
			1&&1\\
			2&&1+(\eta+1)\\
			3&&1+(\eta+1)+(\eta+1)\eta\\
			\vdots&&\vdots\\
			h&&1+(\eta+1)+(\eta+1)\eta+(\eta+1)\eta^2+\cdots+(\eta+1)\eta^{h-2}\\
		\end{align*}
		Hence, the number of nodes $n=1+ (\eta+1)\sum_{i=0}^{h-2}\eta^i$
%
		
	\end{proof}
	

	\begin{figure}[ht]
		\begin{center}
			\scalebox{1}{
				\begin{tabular}{cc}
					\includegraphics[width=0.4\columnwidth]{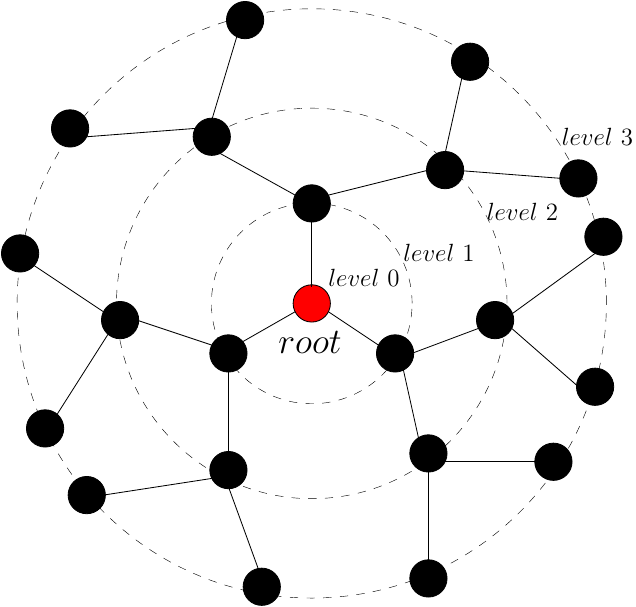} 
					&\includegraphics[width=0.5\columnwidth]{ 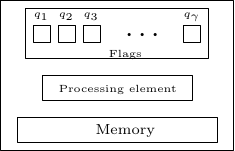}\\
					a.& b.\\				
			\end{tabular}}
			\caption{ Cayley tree; (a) Tree of height $4$, (b) A node of tree}
			\label{tree-cell}
		\end{center}
	\end{figure}
	\noindent The Cayley tree of Figure~\ref{tree-cell}(a) is of $h=4$, $\eta=2$ --that is, $n=22$. While realizing {\em in-memory searching}, {\em computing max (min) in-memory} or {\em in-memory sorting}, the $n$ is to be greater than the number of elements in the input list. For an appropriate $n$, the $h$ is to be properly chosen when the order $\eta$ is given. 
	
	Each node of a Cayley tree contains a tiny processing unit --that is, a memory word and a set of flags. In Figure~\ref{tree-cell}(b), $q_1,~q_2,~q_3,\cdots,q_\gamma$ are such $\gamma$ flags. 
	For the computation at a node, the node communicates with its adjacent (neighbor) nodes. It follows two internal functions $f$ and $f_{in}$. The $f:\{0,1\}^{\eta+\gamma+2}\rightarrow\{0,1\}$ updates one flag (through which the signals and memory bits are sent to designated neighbors) and the $f_{in}:\{0,1\}^{\gamma+1}\rightarrow\{0,1\}^{\gamma}$ updates some other flags and the memory word. Two procedures (send and receive) are defined to explore the communication protocol of a node. The {\em send} procedure sends a bit to its neighbors whereas, the {\em receive} procedure receives bits from its neighbor and perform internal operations ($f$ and $f_{in}$) to update the flags and memory word. For the proposed in-memory computing schemes, the elements of a list initially distributed among the nodes of the tree (except the {\em root}). 

For the searching, the key is in the {\em root} of the tree. In addition to an element of the list, each node accesses three flags -- {\em state}, {\em start}, and {\em match}. To match the key, the signals flow from {\em root} to leaves through the intermediate nodes. Each child node in turn sends a signal to its parent after receiving the signal ({\em match}) from its children. The searching decision, finally, is taken by considering only the flags of the {\em root} once participation of all the nodes is completed.

In case of {\em computing max (min) in-memory}, $\eta+4$ flags are required. These are the {\em state}, {\em start} and $\eta+2$ {\em link} flags. Each link flag represents whether the connection of a node is enabled or disabled with its neighbor (memory word). During execution, the signals flow from leaves to the {\em root} to update the flags and memory word of a node. That is, each node receives bits from its children, compares these with the {\em MSB} of its memory word to update its next {\em state}. The {\em state} is then sent to its parent. The {\em root} receives the bits from its children and updates its flags as well as the {\em MSB} of the memory word. After every such iteration, the memory word of a node is shifted left (\verb*|Circular-Left-Shift|).

The {\em in-memory sorting} scheme consists of a set of cycles of operations. Each cycle includes the {\em in-memory searching} and {\em computing max in-memory}. After every computation of max, the {\em root} holds the maximum value and then reports --that is, stored in different place. An {\em in-memory searching} is then perform to identify the nodes ({\em non-root} nodes) that contain the maximum value. Such nodes are disabled so than these cannot participate in the next iteration of {\em computing max in-memory}. This process (computing max/report/searching/disabling) repeats until all the nodes are disabled. As, at every cycle a value (maximum value) is reported, the sorted elements are obtained in descending order. 

\section{In-Memory Searching}\label{case1}
This section details out the {\em in-memory searching} scheme in the IMC platform, defined in the earlier section. For searching, the {\em root} node is loaded with the key. It initiates the execution by sending a bit $b$ (say, $1$) of the key to its children. An intermediate node that receives $b$ from its parent updates its flags as per \verb*|receive-search| operation (Procedure~\ref{alg:rcv}). It then forwards $b$ as per \verb*|send-search| operation (Procedure~\ref{alg:send}) to its children. A leaf performs similar task except sending of $b$ to its children. The internal structure of an IMC node with flags, memory word, procedures \verb*|send-search| and \verb*|receive-search|, $f$ and $f_{in}$, and $state_0$ to $state_\eta$ ($state_0$, $state_1,\cdots,state_\eta$ are the states of its neighbors) for the search, is shown in Figure~\ref{cellstruct}. 

	The {\em in-memory searching} operation iterates with two phases, Phase $1$ and Phase $2$, in a node. The task of Phase $1$ and Phase $2$, in the {\em root} and in the intermediate nodes, are summarized below.
	\begin{itemize}
		\item [A.] The root, in Phase $1$, first sends {\em initiate} signal to its children and then the memory bits (key) one by one. In Phase $2$, the {\em root} receives {\em match} signals/bits from its children.
		
		\item [B.] An intermediate node $v$, in Phase $1$, receives {\em initiate} signal from its parent. It then updates its own flags ({\em state}, {\em start} and {\em match}) depending on the memory bit (key bit) received from the parent. After each update $f$ and $f_{in}$, the node $v$ forwards the {\em initiate} signal along with its {\em state} and the key bit to its children. In Phase $2$, for a match with the bit of element (stored in $v$) and the corresponding key bit, the node $v$ sends a {\em match} signal to its parent. If any {\em match} signal is received from its children, the node $v$ forwards the signal to its parent in the next time step. Since a leaf node has no children, such a node can communicate with its parent only.
		
		At the end of the {\em in-memory searching}, the search result can be obtained by sensing the state flag of the root only. If the state flag of the root contains $1$ --this implies, the key is found, the key is not found otherwise.
	\end{itemize}
	
		\begin{figure}[ht]
		\begin{center}
			\scalebox{1}{
				\begin{tabular}{c}
					\includegraphics[width=1\textwidth]{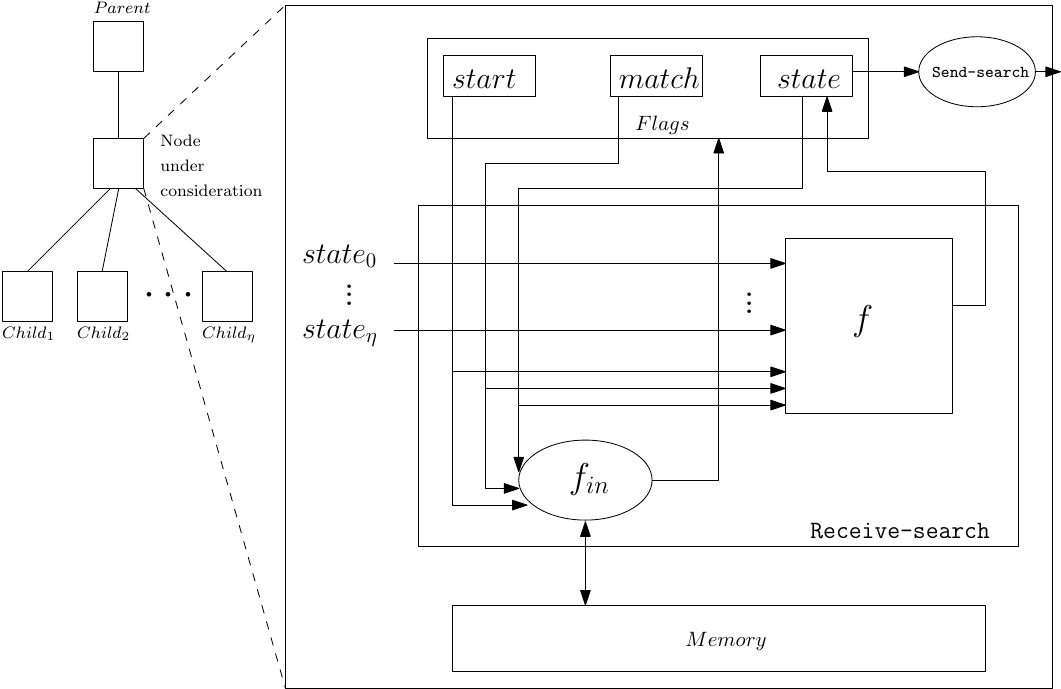}
			\end{tabular}}
			\caption{Internal structure of an IMC node for {\em in-memory searching}}
			\label{cellstruct}
		\end{center}
	\end{figure}
	
	 Procedure~\ref{alg:rcv} and Procedure~\ref{alg:send} are designed to realize the \verb*|receive-search| and \verb*|send-search| operations of each node respectively. Here, in the procedures,
	\begin{itemize}
		\item $\mathscr{B}$ is the memory, and $\mathscr{B}_j[k]$ represents the $k^{th}$ bit of the memory at node $j$.
		\item The $state_j$, $start_j$ and $match_j$ are the three flags of node $j$.	
	\end{itemize}

	\begin{algorithm}[tbh]	
		\begin{algorithmic}
			\fontsize{7.5}{6.5}\selectfont
				\Require A word $\mathscr{B}_j$.
				\State(Local Variable)
				\State $clock_j\gets 0;~ state_j\gets 0;~ start_j\gets 0;~ match_j\gets 1;$
				\State \textbf{Receive: $j$ receives one bit $(state_i)$ from $i$:}
				\State $\mathcal{L}\gets length(\mathscr{B}_{j})$
				\If{$j=root$}
				\If{$clock_j>\mathcal{L}$}
				\If{$state_j=1$}
				\State $state_j\gets state_j$;
				\Else
				\ForAll{$d\in state_i.j$}
				\State $state_j\gets state_j\vee state_d$;
				\EndFor
				\EndIf 	
				\Else
				\State \textbf{wait untill} a bit is received.
				\EndIf
				\Else
				\If{$clock_j\leq\mathcal{L}$}
				\State $state_j\gets state_i$;
				\If{$start_j=0$}
				\If{$state_j=1$}
				\State $start_j\gets 1$;
				\Else
				\State $start_j\gets 0$;
				\EndIf
				
				\Else
				\If{$match_j=1$}
				\If{$state_j=\mathscr{B}_j[clock_j-1]$}
				\State $match_j\gets 1$;
				\Else
				\State $match_j\gets 0$;
				\EndIf
				\Else
				\State $match_j\gets match_j$;
				\EndIf
				
				\EndIf
				\Else
				\State $NS_j\gets 0$;
				\ForAll{$d\in state_i.j$}
				\State $NS_j\gets NS_j\vee state_d$;
				\EndFor
				\State $state_j\gets NS\vee match_j$;
				\If{$match_j=1$}
				\State $match_j\gets 0$;
				\EndIf
				
				\EndIf
				
				\EndIf
				
			
		\end{algorithmic}
		
		\caption{\begin{scriptsize}receive-search in searching of a node $j,~1\leq j\leq n$\end{scriptsize}}\label{alg:rcv}
	\end{algorithm}
	
	\begin{algorithm}[ht]
		\begin{scriptsize}
			\begin{algorithmic}		
				\Require A word $\mathscr{B}_j$.
				\State(Local Variable)
				\State $clock_j\gets 0;~ state_j\gets 0;~ start_j\gets 0;~ match_j\gets 1;$
				\State \textbf{Send: $j$ sends one bit $(state_j)$ to $k$:}
				\State $\mathcal{L}\gets length(\mathscr{B}_{j})$
				\If{$j=root$}
				\If{$clock_j=0$}				
				\State$state_j\gets 1$;
				\ForAll{$d\in state_j.k$}
				\SEND $(state_j)$ to $d$;
				\EndFor
				\ElsIf{$clock_j\leq\mathcal{L}$}
				\State $state_j\gets \mathscr{B}_j[clock_j-1]$;
				\ForAll{$d\in state_j.k$}
				\SEND $(state_j)$ to $d$;
				\EndFor
				\Else
				\State \textbf{wait for} receive bit from $k$;
				\EndIf
				\Else
				\If{$clock_j>\mathcal{L}$}
				\SEND $(state_j)$ to $k$;
				\Else
				\ForAll{$d\in state_j.k$}
				\SEND $(state_j)$ to $d$;
				\EndFor
				\EndIf	
				\EndIf	
				\State$clock_j \gets clock_j+1$;	
				
			\end{algorithmic}
		\end{scriptsize}
		\caption{\begin{scriptsize}
				send-search in searching of a node $j,~1\leq j\leq n$
		\end{scriptsize}}\label{alg:send}
	\end{algorithm}
	
	\noindent Figures~\ref{Phase1A} to~\ref{Phase2A} explain the {\em in-memory searching}. Here, order of the tree $\eta=2$ -that is, a node has three adjacent nodes, and 
	nine $4$-bit elements are in the list. The elements are distributed over the $9$ nodes (except the root). The elements are $14$, $9$, $6$, $10$, $14$, $7$, $11$, $11$ and $10$. The key $9$ is loaded in memory of the root. Initially, all the flags in {\em root} are set to $1$ and for the remaining nodes, the \emph{state}, \emph{start} and \emph{match} are set to $0$, $0$ and $1$ respectively (Figure~\ref{Phase1A}(a)). The Phase 1 and Phase 2 of the search operation are next illustrated.

	\noindent\textbf{\underline {Phase 1}}
	
	Figures~\ref{Phase1A}(a)-\ref{Phase1A}(c) are to describe the initial iterations of Phase $1$, where each node sends its {\em state} (marked as blue arrow) to its children. After the {\em initiate} signal, the {\em state} flag of the root assumes the value of {\em MSB} of the key, which is sent to the children of the root. Then, the next bit of the key is assigned to the {\em state} flag, and is sent to the children. This process repeats until all the bits of the key are sent.

	\begin{figure}[hbt!]
		\begin{center}
			\scalebox{1}{
				\begin{tabular}{ccc}
					\includegraphics[width=0.45\textwidth]{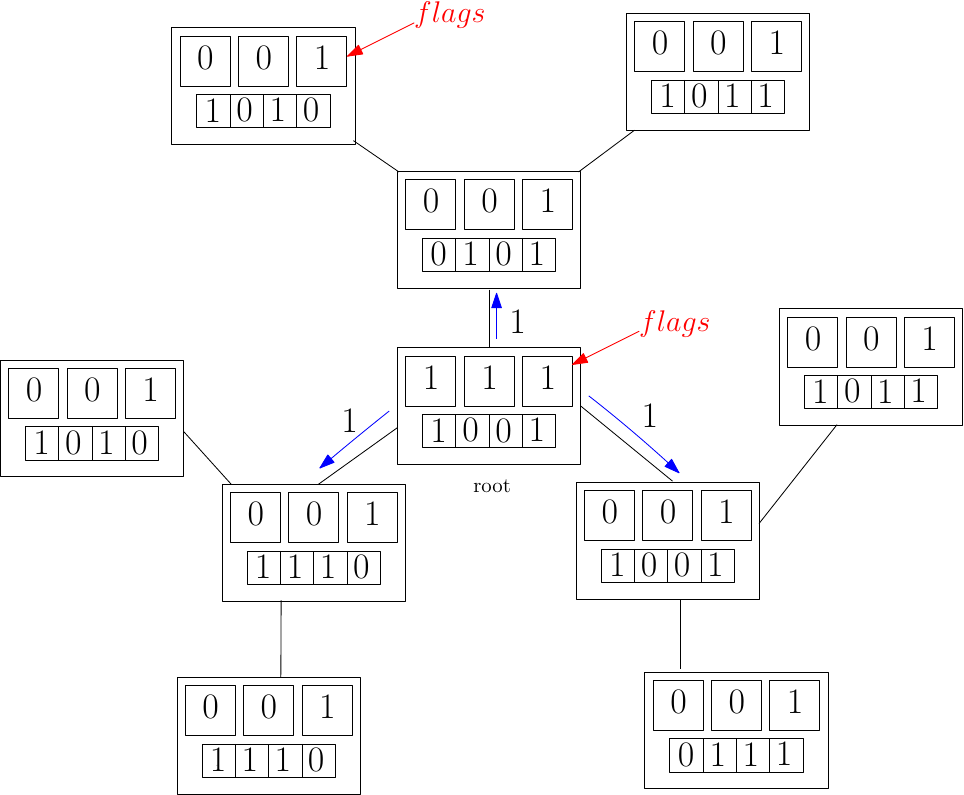}&&
					\includegraphics[width=0.45\textwidth]{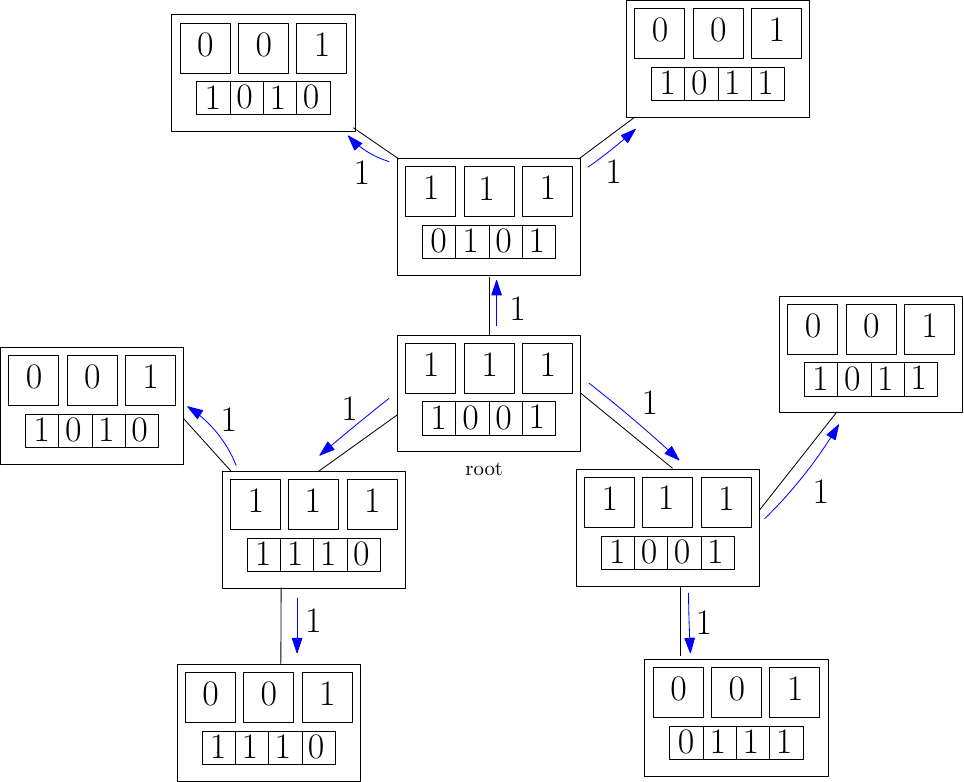}\\
					a.&&b.\\
			\end{tabular}}
			
			\scalebox{1}{
				\begin{tabular}{ccc}
					&\includegraphics[width=0.45\textwidth]{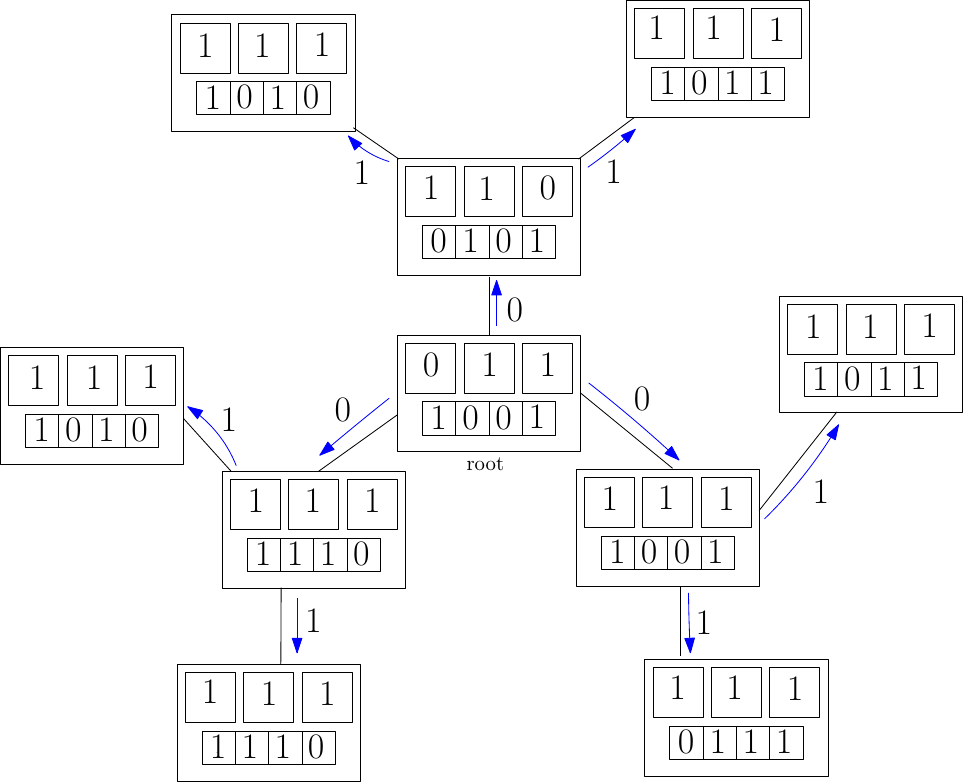}&\\
					&c.&\\
			\end{tabular}}
			\caption{Three consecutive steps of Phase $1$ from the initial configuration}
			\label{Phase1A}
		\end{center}
	\end{figure}
	
	During Phase~$1$, each intermediate node receives an {\em initiate} signal from its parent, and its {\em start} flag is set to $1$. A node then sends the {\em initiate} signal to its children. Next, the node receives a memory-bit from its parent. If the bit is similar to the MSB of its memory, its {\em match} flag is set to $1$ otherwise, the {\em match} flag is set to $0$. During this phase, if {\em match} flag is set to $0$ at any iteration, it sticks to $0$ for the rest of the iterations (Figure~\ref{Phase1A}(b) and \ref{Phase1A}(c)).
	
	Once the match flag is set (set to $0$ or $1$), the received bit is then sent to its children (see Figure~\ref{Phase1B}). After that, the node receives the next memory bit from its parent and compares it with the second {\em MSB} of its memory word and so on. These received bits are sent to its children. That is, through this process, an intermediate node receives the whole key and compares it with the element stored in its memory. 
	
	\begin{figure}[hbt!]
		\begin{center}
			\scalebox{1}{
				\begin{tabular}{ccc}
					\includegraphics[width=0.45\textwidth]{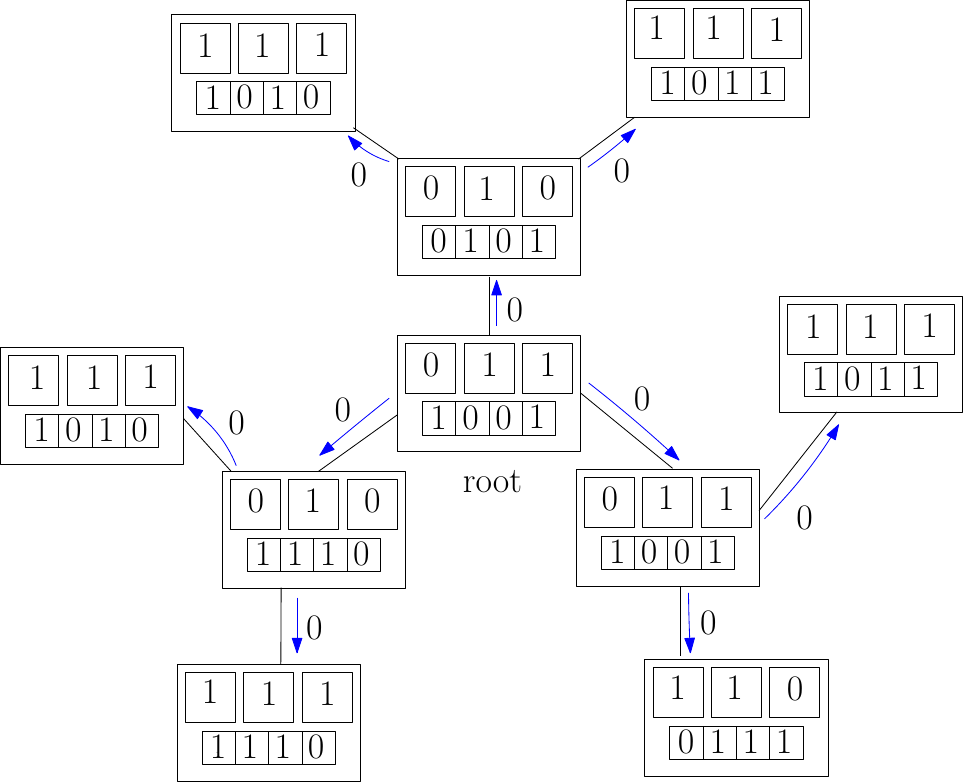}&&
					\includegraphics[width=0.45\textwidth]{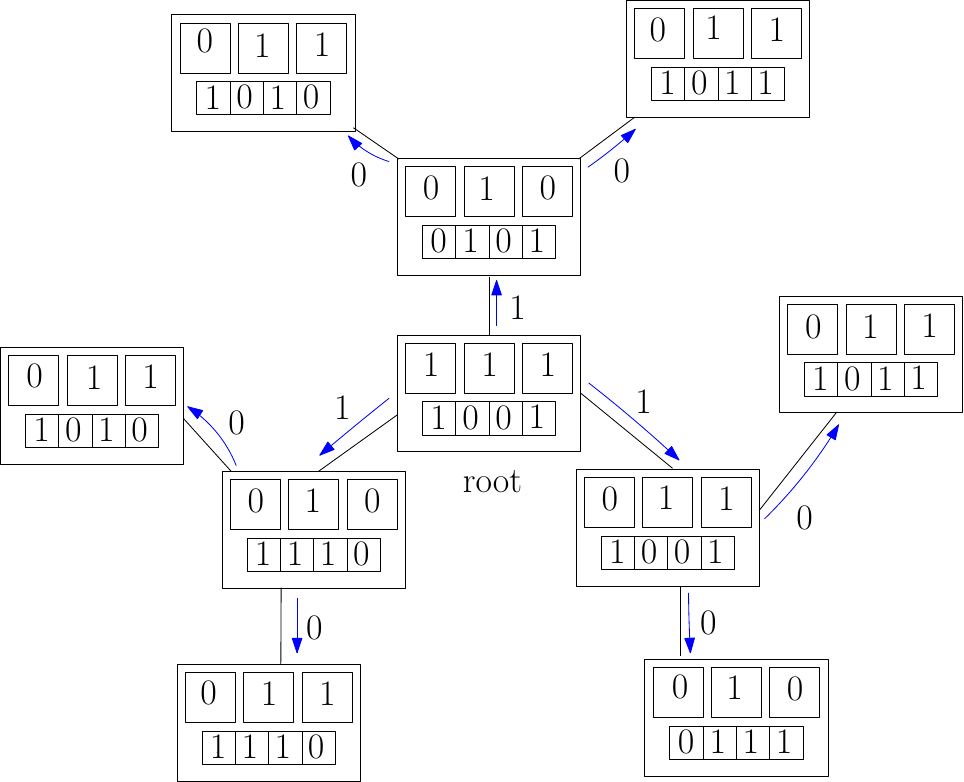}\\
					a.&&b.\\
			\end{tabular}}
			
			\scalebox{1}{
				\begin{tabular}{ccc}
					&\includegraphics[width=0.45\textwidth]{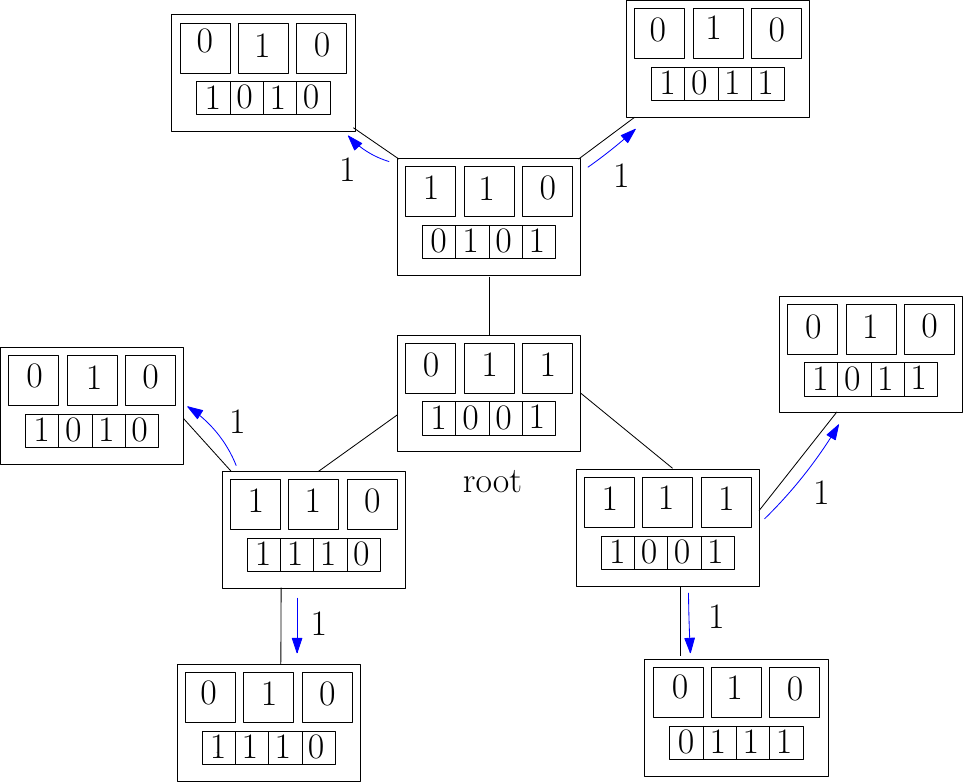}&\\
					&c.&\\
			\end{tabular}}
			
			\caption{Steps of Phase $1$ of intermediate nodes}
			\label{Phase1B}
		\end{center}
	\end{figure}

\noindent\textbf{\underline {Phase 2}} 

In Phase~$2$, at first the {\em state} flag of each intermediate node is replaced by the {\em match} flag. Then the {\em state} flag is sent to its parent. 
In Figure~\ref{Phase2B}, the {\em match} flag of each node is assigned to its {\em state} flag and then the {\em match} flag is set to $0$ (Figure~\ref{Phase2B}(b) and Figure~\ref{Phase2B}(c)). A leaf node, on the other hand, first switches its {\em state} according to the {\em match} flag and then sends the {\em state} flag to its parent (Figure~\ref{Phase2B}(a) -- \ref{Phase2B}(c)).

\begin{figure}[hbt!]
	\begin{center}
		\scalebox{1}{
			\begin{tabular}{ccc}
				\includegraphics[width=0.45\textwidth]{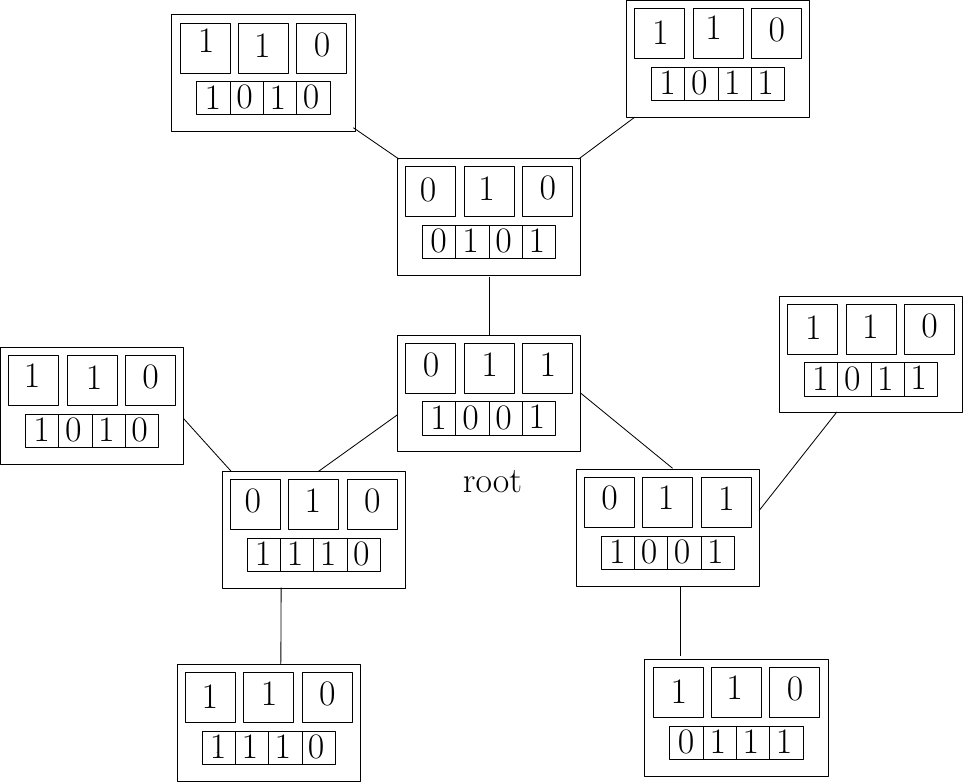}&&
				\includegraphics[width=0.45\textwidth]{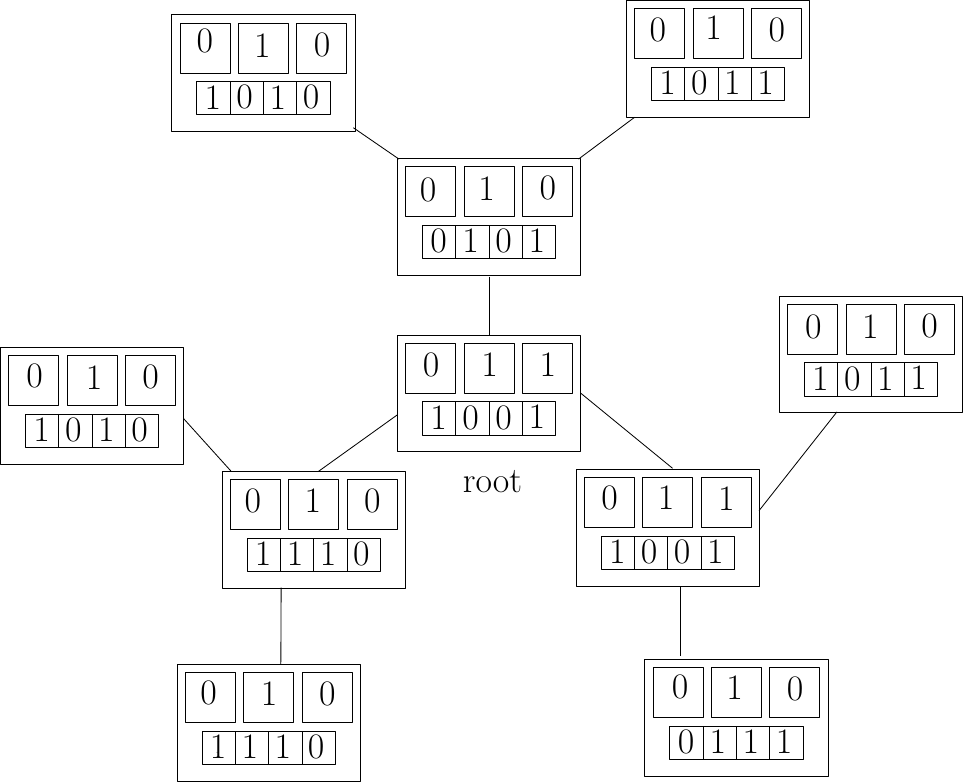}\\
				a.&&b.\\
		\end{tabular}}
		
		\scalebox{1}{
			\begin{tabular}{ccc}
				&\includegraphics[width=0.45\textwidth]{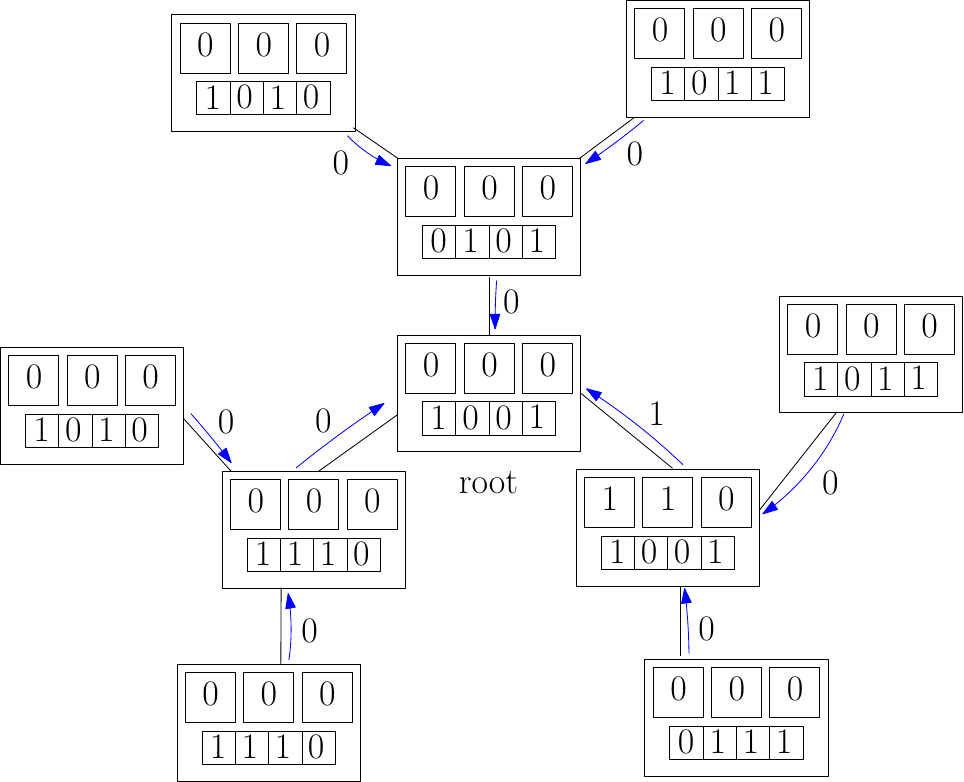}&\\
				&c.&\\
		\end{tabular}}
		\caption{Phase $2$ -- starts from b}
		\label{Phase2B}
	\end{center}
\end{figure}


Figure~\ref{Phase2A}(a) shows that the {\em root} node starts receiving the state bits from its children. The {\em state} flag of the {\em root} is switched to $1$ in Figure~\ref{Phase2A}(b) because it receives $1$ from one of its children. Once the {\em state} of the {\em root} is switched to $1$, it sticks to $1$ (Figure~\ref{Phase2A}(a) --~\ref{Phase2A}(c)). At the end of execution, since the state of the root contains 1 in Figure~\ref{Phase2A}(c), the key is found. 

\begin{figure}[hbt!]
	\begin{center}
		\scalebox{1}{
			\begin{tabular}{ccc}
				\includegraphics[width=0.45\textwidth]{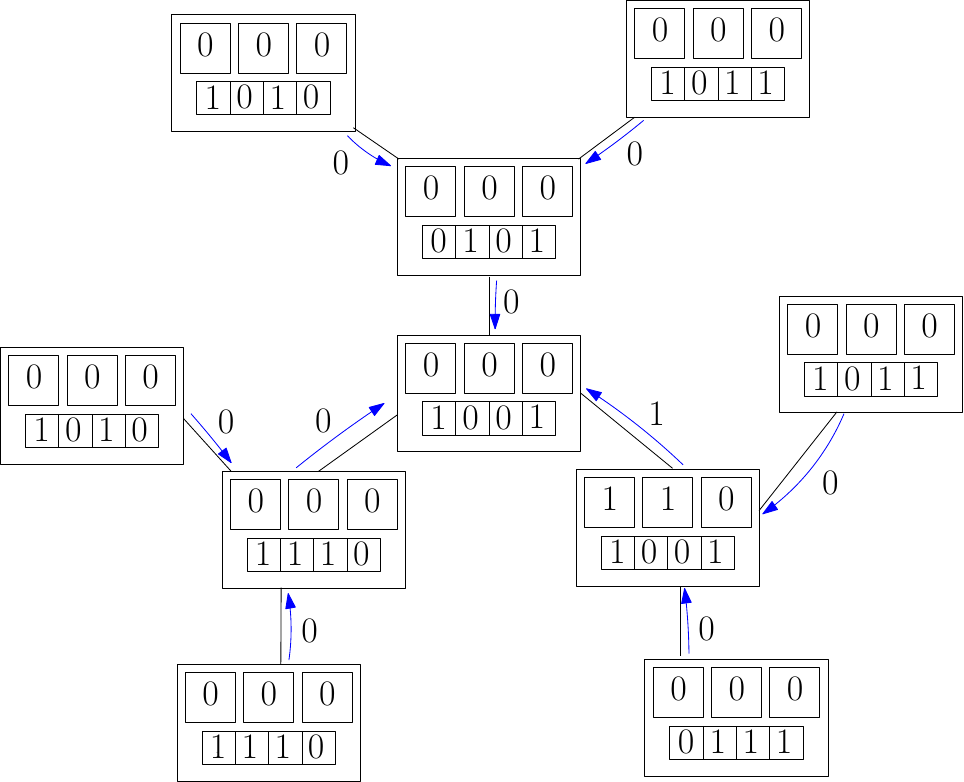}&&
				\includegraphics[width=0.45\textwidth]{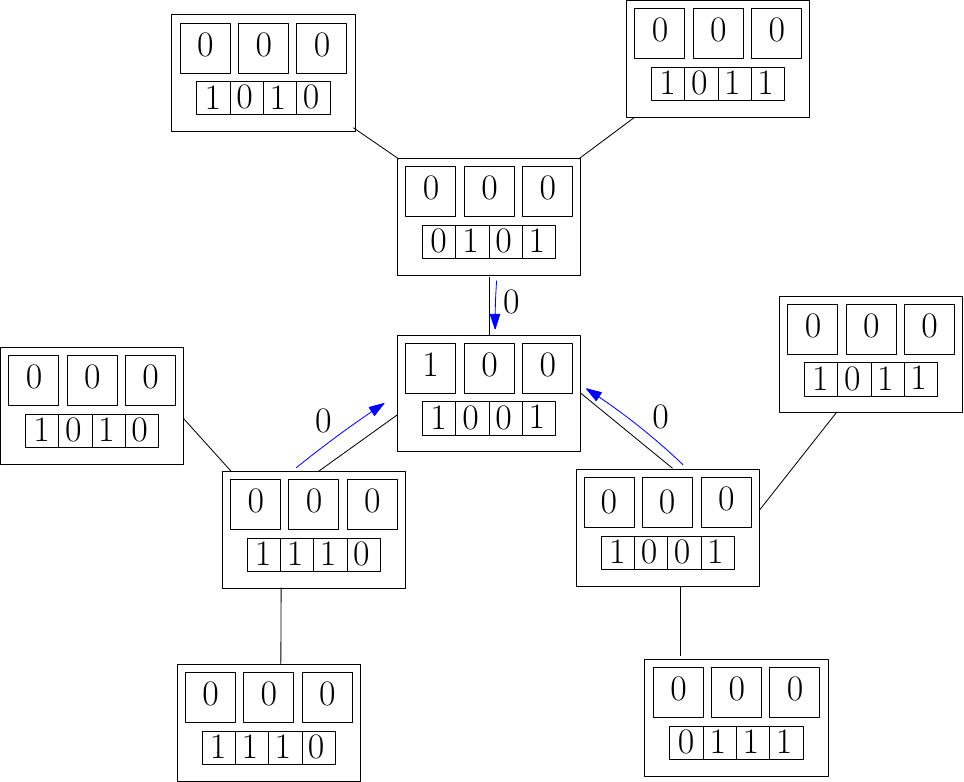}\\
				a.&&b.\\
		\end{tabular}}
		
		\scalebox{1}{
			\begin{tabular}{ccc}
				&\includegraphics[width=0.45\textwidth]{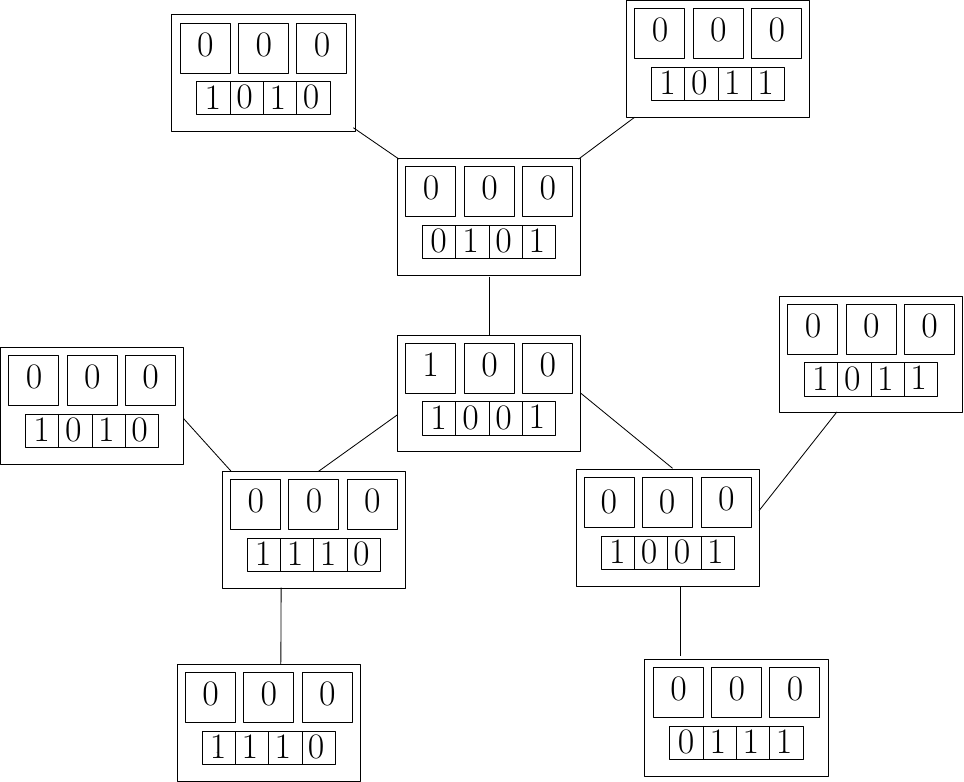}&\\
				&c.&\\
		\end{tabular}}
				%
		\caption{Phase $2$ -- the search terminates at c}
		\label{Phase2A}
	\end{center}
\end{figure}

\vspace{1 em}
\begin{proposition}
	The proposed in-memory searching scheme terminates in finite time for any finite input list. 
\end{proposition}
\begin{proof}
	
	For any finite input list, Phase $1$ of the scheme allows the spread of initiate signal and memory bits, after which the scheme essentially enters into Phase $2$. In Phase $2$, each node first copies the {\em match} flag to the {\em state} flag and then sends it to its parent so that the {\em match} signal reaches the root. Once the {\em match} signal is sent, a node switches its match to $0$. Now, there can be two possible cases-- 
	\begin{itemize}
		\item Let at the end of Phase $1$, the {\em match} flags of all the nodes be $0$. Then, no match signal (through the {\em state} flag) reach the root. Hence, the {\em state} of the {\em root} remains $0$. The Phase $2$ reaches to a scenario where all the {\em state} flags of the nodes are $0$. 
		
		\item At the end of Phase $1$, there exist some {\em match} flags in the Cayley tree nodes that contain $1$. Hence, at the end of Phase 2, the {\em match} signal reaches (through the {\em state} flag) the root, and the {\em state} of the {\em root} switches to $1$, and it remains $1$. The Phase $2$ reaches as all but one {\em state} flag is $1$.
	\end{itemize} 	
	In both the cases, the scheme terminates.
\end{proof}


\begin{proposition}\label{prop3}
	Time complexity of the Cayley tree based {\em in-memory searching} is $\mathcal{O}(\log{}n)$.
\end{proposition}
\begin{proof}
	Let assume that the word size of the memory is $w$ and the height of the tree is $h$. The Phase~1 is performed in $w+h$ time steps to send the {\em initiate} signal and every bits of the key to all the nodes of the tree. On the other hand, Phase $2$ is computed for $h$ iterations to send the {\em match} signal to the root. Hence, after $w+2h$ time steps, the IMC platform can decide on the outcome of the searching by sensing the {\em root} only. Since the word size $w$ and the number of flags are constant for the IMC model, the execution time depends on the height $h$ only. That is, the complexity of the search in the IMC platform is $\mathcal{O}(h)$. Since the height $h=\mathcal{O}(\log{n})$ (Proposition~\ref{th1}), the time complexity of the proposed {\em in-memory searching} scheme is $\mathcal{O}(\log{n})$.
\end{proof}


Since each Cayley tree node in the {\em in-memory searching} scheme assumes memory word and three flags. That is, apart from the elements, the additional space required for the searching is $3n$ bits.

\section{Computing Max (Min) In-memory}\label{sec-max}
This section introduces the {\em computing max (min)} of an input list in the IMC platform. As in Section~\ref{case1}, the elements of a list are distributed  in the nodes (except the {\em root}) of the Cayley tree of IMC platform. The {\em root} is loaded with $0$. At the end of the computation, the of the {\em root} contains the max (min) element.

Here, each node of the tree contains $\eta+4$ flags in addition to a memory word. Let consider that the order of the tree $\eta=2$, then the flags are {\em state}, {\em start}, {\em link of memory} ($l_{m}$), {\em link of left child} ($l_{c_1}$), {\em link of right child} ($l_{c_2}$) and {\em link of parent} ($l_{p}$). Since the root has no parent, the $l_p$ for root corresponds to its middle child. 

For computation of {\em max}, the {\em state} and {\em start} of each leaf node are set to $1$, and the {\em state} and {\em start} of other nodes of the tree are set to $0$. The other flags, i.e., $l_{m},~l_{c_1},~l_{c_2}$ and $l_{p}$ of each node of the tree are set to $0$. Then each leaf sends its {\em state} to its parent as the {\em initiate} signal. The other nodes, after receiving the {\em initiate} signal, update their flags and forward the {\em initiate} signal to their parents. From the second iteration onward, a bit (at MS digit position) of the memory word of each leaf is assigned to its {\em state} flag. It then sends the updated {\em state} flag to its parent and performs an one bit \verb*|circular-left-shift| to its memory. This process repeats till all the memory bits are scanned and sent to its parent. 

During this process, when an intermediate node first receives the initiate signal, its {\em start} flag is set to $1$ and the initiate signal is forwarded to its parent. After that, the node receives bits from each child if the corresponding link flag is $0$. Similarly, the node take the {\em MSB} of its memory word if the corresponding flag $l_m$ of its memory is $0$. It then perform logical \verb*|OR| operation on all the bits (received bits and {\em MSB}). Note that, if any link flag say, $l_c$ (for a child say, $c$) or $l_m$ (for its memory word) is $1$, the received bit from the child $c$ or the {\em MSB} is not considered for the \verb*|OR| operation. The outcome of the \verb*|OR| operation is then assigned to its state.
The {\em state} is then sent to its parent and \verb*|circular-left-shift| is done on its memory word.

Once the root receives the initiate signal, it sets its {\em start} flag to $1$. From the next iterations, it receives a bit from each child, say $c_1$, if the corresponding link of the child that is, $l_{c_1}$ is $0$. That is, the root receives $\eta+1$ bits from its children if all its child link flags --that is, $l_{c_1}$, $l_{c_2}$, $\cdots$ $l_{c_{\eta}}$ and $l_{c_p}$ contain $0$. It then perform a logical \verb*|OR| operation of the received bits. The out come of the \verb*|OR| operation is stored in the {\em MSB} of its memory. After this operations, the root performs a \verb*|circular-left-shift| operation on its memory word. At the end of the scheme, the memory word of the root contains the {\em max} element of the input list.
The earlier discussion points to the fact that the computation of max is done through execution of two tasks (procedures) within a node. These are the \verb*|send-max| (Procedure~\ref{alg:send_max}) and \verb*|receive-max| (Procedure~\ref{alg:rcv_max}). 
Here, in the procedures,
\begin{itemize}
	\item $\mathscr{B}_j$ is the memory of node $j$.
	\item $state_j$ and $start_j$ are the {\em state} and {\em start} flags of a node $j$. Additionally, one memory flag say, $l_m$ and $\eta+1$ link flags for its $\eta+1$ neighbors are there. A link flag $l_x$ of a node corresponds to the neighbor $x$ of the node. A node does not consider the received bit from $x$ if $l_x$ is $1$.
\end{itemize}
\begin{algorithm}[ht]
	\begin{scriptsize}
		\begin{algorithmic}		
			\Require A word $\mathscr{B}_j$.
			\State(Local Variable)
			\State $clock_j\gets 0;~ state_j\gets 0;$
			\State \textbf{Send: $j$ sends one bit $(state_j)$ to $k$:}
			
			\If{$j\in leaves$}
			\If{$clock_j$ = $0$}
			\State $state_j \gets 1$
			\Else
			\State $state_j \gets \mathcal{MSB}.\mathscr{B}_j$
			
			\EndIf
			\ForAll{$d\in state_j.k$}
			\If{$d\in state_j.parent$}
			\SEND $(state_j)$ to $d$;
			\EndIf
			\EndFor
			\State \verb*|Circular-Left-Shift| on $\mathscr{B}_j$	
			\ElsIf {$j\in intermediate-nodes$}
			\ForAll{$d\in state_j.k$}
			\If{$d\in state_j.parent$}
			\SEND $(state_j)$ to $d$;
			\EndIf
			\EndFor
			\EndIf	
			
			\State$clock_j \gets clock_j+1$;	
			
		\end{algorithmic}
	\end{scriptsize}
	\caption{\begin{scriptsize}
			send-max operation in {\em computing max} of a node $j,~1\leq j\leq n$
	\end{scriptsize}}\label{alg:send_max}
	
\end{algorithm}

			%
			%
	%

\begin{algorithm}[ht]
	\fontsize{8.5}{7.5}\selectfont
		\begin{algorithmic}
			
			\Require A word $\mathscr{B}_j$.
			\State(Local Variable)
			\State $clock_j\gets 0;~ state_j\gets 0;~ start_j\gets 0;~ l_{\mathscr{B}_j}\gets0;$
			\ForAll{$\mathcal{Y}\in neighbours.j$}
				\State $l_\mathcal{Y}\gets0$;
			\EndFor
			
			\State \textbf{Receive: $j$ receives one bit $(state_i)$ from $i$:}
			\If{$j=root$}
			\If{$start_j=0$}
			
			\State \textbf{wait until} a bit is received.
			\Else
			\ForAll{$d\in state_i.j$ and $d\neq state_{i(parent)}.j$ and $l_d=0$} 
			\State $s\gets s\vee state_d$;
			\If{$state_d\neq s$}
			\State $l_d\gets 1$;
			\EndIf
			\EndFor	
			
			\State $state_j\gets s$
			
			\State $\mathcal{MSB}.\mathscr{B}_j\gets state_j$
			\State \verb*|Circular-Left-Shift| on $\mathscr{B}_j$						
			\EndIf
			
			\ElsIf {$j\in intermediate-nodes$}
			\If{$start_j=0$}
			\State $state_j\gets state_i$.
			\Else
			\ForAll{$d\in state_i.j$ and $d\neq state_{i(parent)}.j$ and $l_d=0$}
			\State $s\gets s\vee state_d$;
			\If{$state_d\neq s$}
			\State $l_d\gets 1$;
			\EndIf
			\EndFor
			\If{$l_{\mathscr{B}_j}=0$}
			\State $s\gets s\vee \mathcal{MSB}.\mathscr{B}_j$;
			\If{$\mathcal{MSB}.\mathscr{B}_j=s$}
			\State $l_{\mathscr{B}_j}\gets 0$
			\Else
			\State  $l_{\mathscr{B}_j}\gets 1$
			\EndIf
			\EndIf	
			
			\State $state_j\gets s$
			\State \verb*|Circular-Left-Shift| on $\mathscr{B}_j$
			\EndIf
			
			
			\EndIf
			\State $clock_j \gets clock_j+1$;
			
		\end{algorithmic}
	\caption{\begin{scriptsize}receive-max operation in {\em computing max} of a node $j,~1\leq j\leq n$\end{scriptsize}}\label{alg:rcv_max}
\end{algorithm}


\begin{figure}[hbt!]
	\begin{center}
		\scalebox{1}{
			\begin{tabular}{ccc}
				\includegraphics[width=0.45\textwidth]{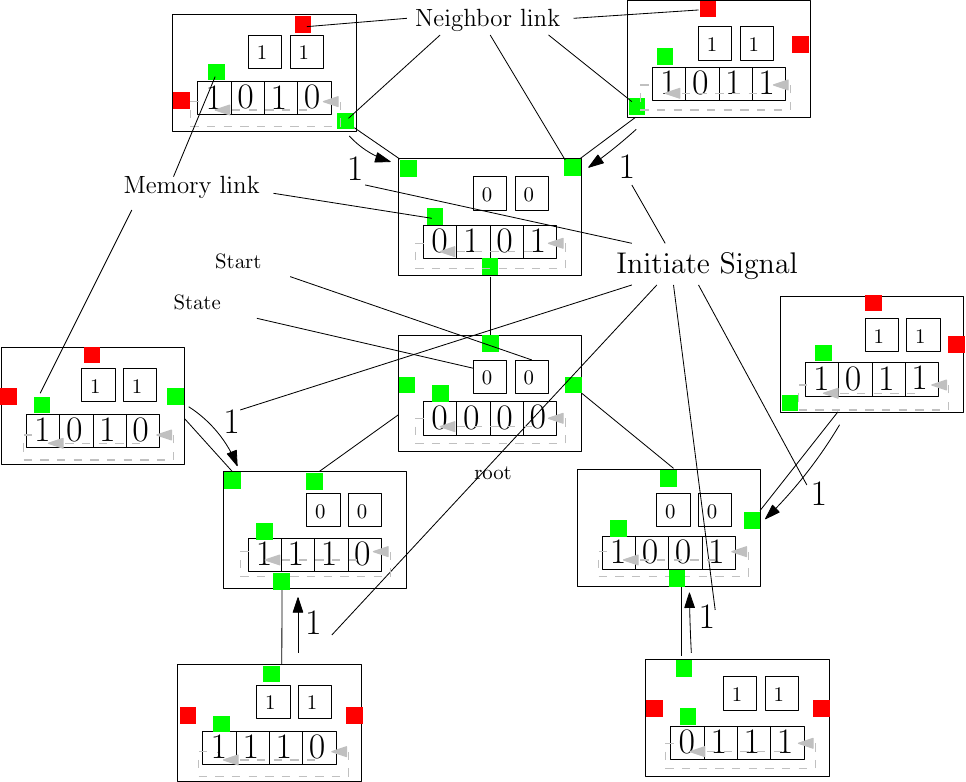}&&
				\includegraphics[width=0.45\textwidth]{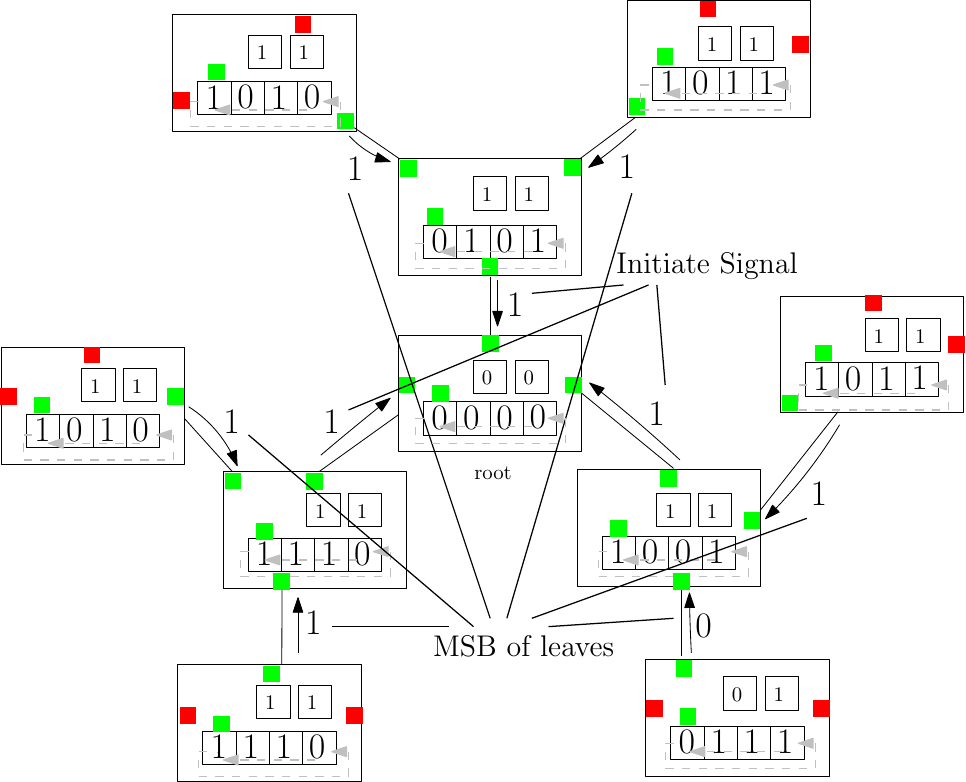}\\
				(a)&&(b)\\
		\end{tabular}}
	\scalebox{1}{
		\begin{tabular}{ccc}
			\includegraphics[width=0.45\textwidth]{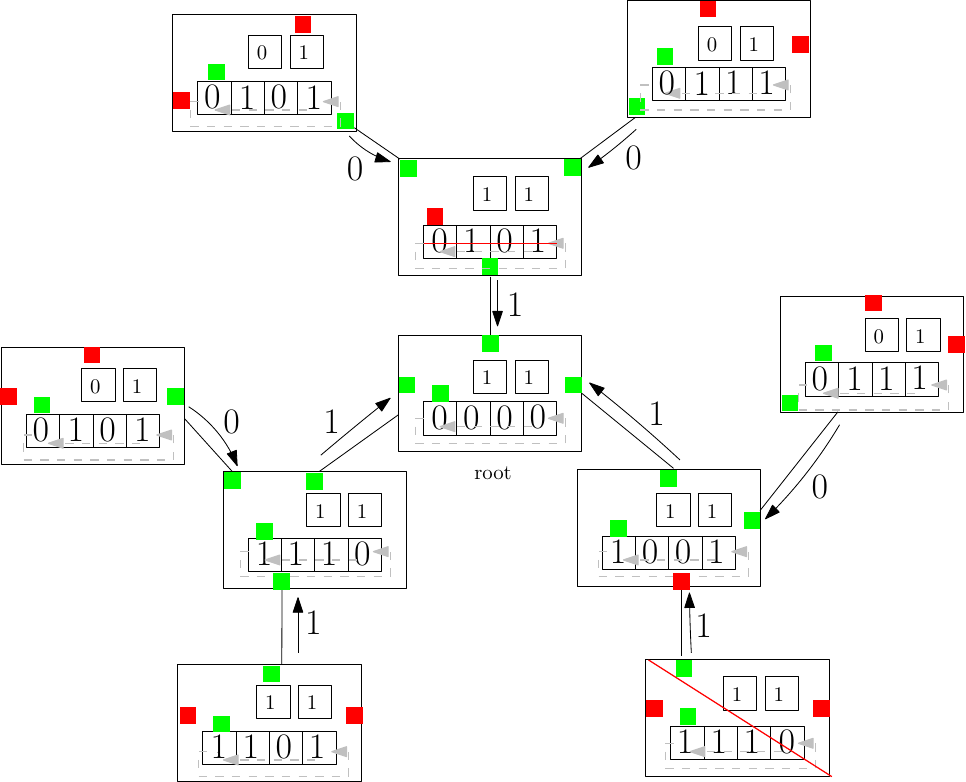}&&
			\includegraphics[width=0.45\textwidth]{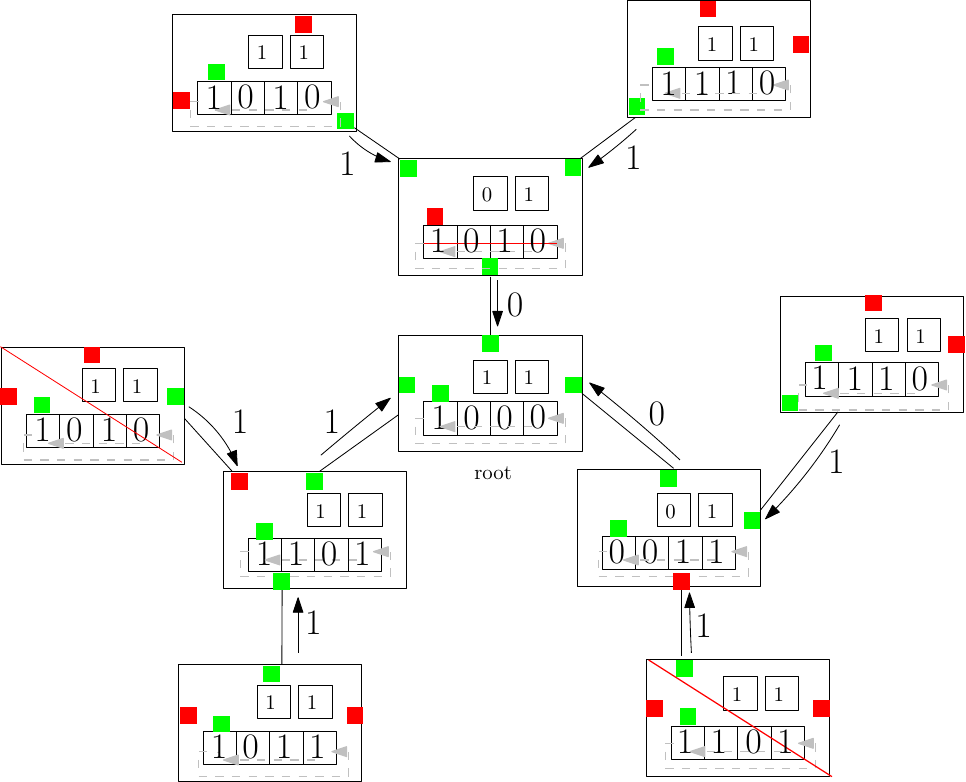}\\
			(c)&&(d)\\
	\end{tabular}}
		\scalebox{1}{
			\begin{tabular}{ccc}
				\includegraphics[width=0.45\textwidth]{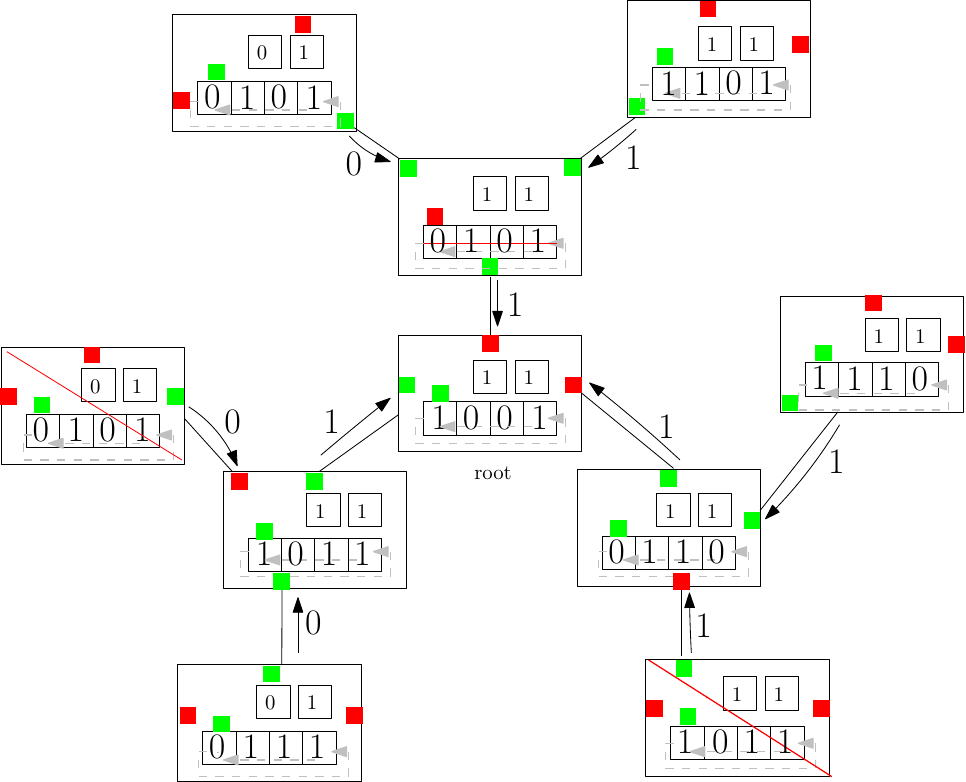}&&
				\includegraphics[width=0.45\textwidth]{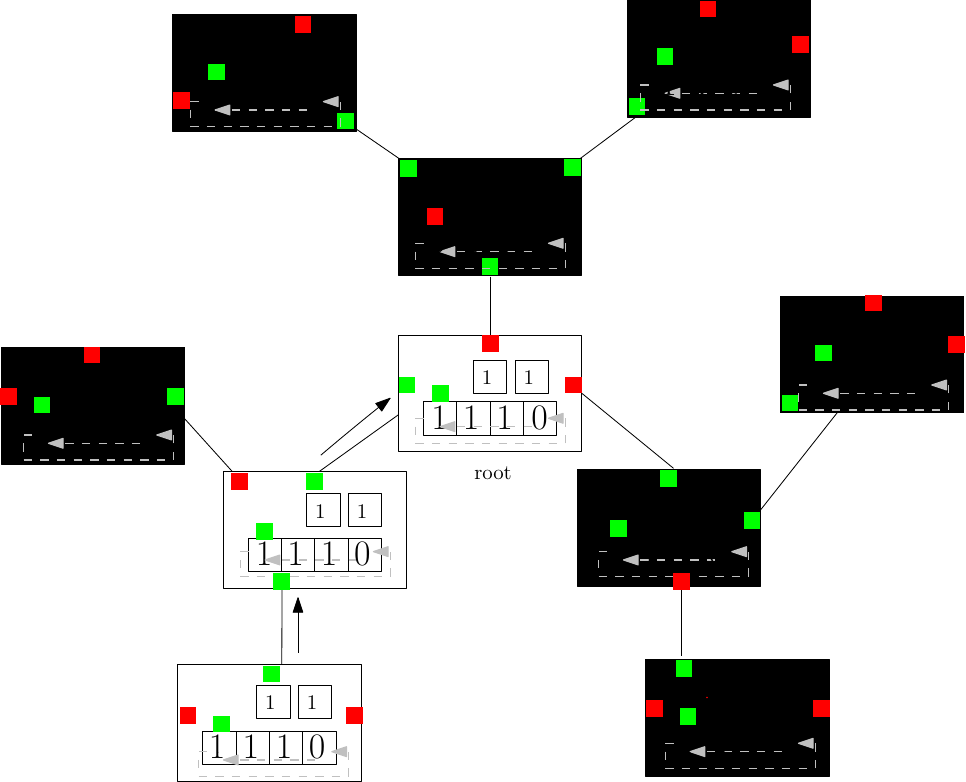}\\
				(e)&&(f)\\
		\end{tabular}}
				%
		\caption{The implementation of {\em computing max in-memory} with a Cayley tree of height $3$. Initially, the elements are $14,9,5,14,7,11,10$ and $10$, which are distributed among the nodes. The memory word of the root contains $0$. (a) The initialization of flags; (f) The root contains the max element}
		\label{fig:max1}
	\end{center}
	\vspace{-2 em}
\end{figure}

Figure~\ref{fig:max1} shows an example of {\em computing max in-memory}, where the order of the tree $\eta=2$, and $9$ nodes (except the root) are there in the tree, where the memory word of each node contains a $4$-bit element. The elements are $14,~9,~5,~14,~7,~11,~10$ and $10$, which are distributed among the nodes. The memory word of the root contains $0$. Figure~\ref{fig:max1}(a) shows the initialization of each node of the tree for execution, where the state, start, link of the memory flag, and the link of the child flags are marked. The green boxes (resp. red boxes) in a node indicate the link flags representing the flag with value $0$ (resp. $1$).

In Figure~\ref{fig:max1}(a), the leaves start the execution by sending an initiate signal to their parents. The leaves then send the MSB of their memory word and perform a circular shift on their memory word, whereas the intermediate nodes hold the initiate signal to its state, switch their start flag to $1$, and forward the initiate signal to the root (Figure~\ref{fig:max1}(b)). After that, in each iteration, each leaf sends its MSB to its parent and performs a circular left shift on its memory word (Figure~\ref{fig:max1}(b)-Figure~\ref{fig:max1}(e)). After sending the initiate signal, an intermediate node starts receiving MSBs of its children's memory words. It then performs a logical \verb*|OR| operation on the received bits; for example, the bottom right child of the root in Figure~\ref{fig:max1}(b) receives $0$ and $1$ from its children; the MSB of its memory word contains $1$, i.e., its state switches to $1$ (i.e., $0\vee 1\vee 1$) in its next iteration. Since one child sent $0$, its link flag is set to $1$ --that is, the link flag is switched to red (Figure~\ref{fig:max1}(c)). The node does not consider this child (the child is disabled (marked using a red line)) in further iterations (Figure~\ref{fig:max1}(c)). Similarly, the upper child of the root in Figure~\ref{fig:max1}(c) disables its memory by setting the memory link to red ($l_m$ is set to $0$). 

The detailed execution of the given example is shown in Figure~\ref{fig:max1}. In Figure~\ref{fig:max1}(f), the execution ended, and the root contains the max value. Since two child links of the root become red during execution, the corresponding nodes (sub-trees) that are connected with these red links are marked as black. The black node indicates that the node is disabled and does not contribute to the execution. A similar thing happens with one leaf node on the left in Figure~\ref{fig:max1}(f).

\emph{Complexity}: Similar to Proposition~\ref{prop3}, let assume that a Cayley tree contains $n$ nodes, the word size is $w$ and the height of the tree is $h$. To compute max, at first the {\em initiate} signal traverses from leaf to the root. Then each node (except the root) sends the {\em MSB} of its memory to its parent and performs circular left shift on memory. Now, one iteration is needed to traverse the initial signal, and a maximum of $w+h$ iterations are needed to transfer each bit of the memory word from a leaf node to the root. However, each non-leaf node processes the received bit and sends it to its parent in constant time. This implies, the {\em computing max} requires $\log{n}+w+1$ steps. As $w$ is constant and much lesser than $\log{n}$, the time complexity of computing max in-memory is $\mathcal{O}(\log{n})$.





The {\em computing min} scheme, on the other hand, is similar to that of {\em computing max} except the logical operation defined for computing max. In case of {\em computing min}, the logical \verb*|AND| operation is to be done, instead of \verb*|OR|. Hence, a node obtains the minimum bit by performing logical \verb*|AND| among the received bits and {\em MSB} of memory word. It then, disables a link $l_{c_1}$ if the {\em state} of the child $c_1$ is not equal to the result of the \verb*|AND| operation. Similarly, it disables the memory link $l_{m}$ if the {\em MSB} of its memory word is not equal to the result of the \verb*|AND| operation.

\section{In-Memory Sorting}\label{sec-sorting}
The {\em in-memory sorting} scheme is developed following the {\em in-memory searching} and {\em computing max (min)} schemes introduced in Section~\ref{case1} and Section~\ref{sec-max} respectively. Let the elements of the input list are distributed among the nodes of the Cayley tree and the memory of {\em root} is loaded with $0$, where $w$ be the word size of memory and $h$ is the height of the tree. The first max can be obtained in {\em root} node's memory after the $(w+h+1)$ steps (Section~\ref{sec-max}). From this and onwards, after every $2(w+h+1)$ steps (that includes searching an element and computing max), the {\em root} contains the next max element. Hence, the content of memory word of {\em root}, after every $2(w+h+1)$ steps, can be reported (stored) in different memory locations and, thereby, the sorted list is found. The number of elements with same value are stored in consecutive memory locations using the standard procedure of in-memory copy / in-memory bulk copy.

For {\em in-memory sorting}, the required flags are the {\em state}, {\em start}, {\em match}, and $\eta+2$ link flags --that is, $l_{c_1}$, $l_{c_2}$, $l_{p}$ and $l_{m}$ for the $\eta=2$ Cayley tree. The {\em in-memory sorting} effectively consists of two phases for each cycle. {\em Phase A} identifies the max element that can be found in the root. {\em Phase B}, on the other hand, identifies the nodes that contain the max element and disables the nodes. The detail of Phase A and Phase B follows. 

 \paragraph{Phase A} The flags of all the nodes are set as in the computing max scheme, and the {\em match} flag is set to $1$. At first the computing max is done (Section~\ref{sec-max}) so that the {\em root} contains the max element. Once max is in root, the element is reported (stored). Then, the flags of all the nodes are reset to enable {\em search} operation (Section~\ref{case1}). During this reset operation, all the links which are disabled during this phase are also enabled.
	
  \paragraph{Phase B}  In {\em Phase B}, {\em Phase~1} of the search operation (Section~\ref{case1}) is performed to identify the nodes that contain the same element as in the root. Next, the $l_{m}$ of the identified nodes are disabled permanently so that the element of the node cannot further participate in the sorting process. Contents of these nodes (max elements) along with the content of the root are stored (reported) in consecutive memory locations using standard procedure of in-memory copy / in-memory bulk copy. Such permanently disabled links cannot be enabled during any {\em reset} operation. Now, the flags of all the nodes are reset for the next {\em Phase A} --that is, computing the next max element.

These two phases repeat during the {\em in-memory sorting} until all the $l_{m}$'s of the nodes are permanently disabled. Procedure~\ref{alg:in-memory-sorting} provides the algorithmic steps for the {\em in-memory sorting}. It is to be noted that, the number of iterations (Phase A and Phase B) is same as the number of unique elements in the input list.

\begin{proposition}\label{in-memory-sort-p1}
	Time complexity of the Cayley tree based {\em in-memory sorting} is $\mathcal{O}(n\log{n})$.
\end{proposition}
\begin{proof}
	Let the word size be $w$, and the height of the Cayley tree be $h$. The Phase~A is similar to the {\em computing max (min)} scheme. Hence, it performs in $w+h+1$ time steps to send the {\em initiate} signal and $w$ bits (memory word) from leaves to the {\em root} in the Cayley tree. At the end of this phase, the {\em root} contains the max element. An extra time step is then required to reset the flags.
	
	On the other hand, Phase B is similar to Phase 1 of the {\em in-memory searching} of Section~\ref{case1}. Therefore, it requires $w+h+1$ time steps. Then it disables the $l_{m}$s of the nodes that contain the max element and reset the flags of all the nodes of the tree for {\em computing} the next max (min) element (next Phase A).
	
	The above processes repeat $n$ times, where $n$ is the number of unique elements in the input list. Since the word size $w$ and the number of flags are constant, the execution time depends on the $h$ only. That is, the time complexity of the {\em in-memory sorting}, in the IMC platform, is $\mathcal{O}(nh)$ --that is, $\mathcal{O}(n\log{n})$.
\end{proof}



The proposed sorting scheme consists of searching and computing max. Hence, the required number of flags for the scheme is $\eta+5$. That is, the additional space requirement is $n(w+\eta+5)$.


\begin{algorithm}[ht]
	\begin{scriptsize}
		\begin{algorithmic}			
			\State \textbf{Step 1:} Find the max (min) value using {\em computing max (min)} (Procedure~\ref{alg:send_max} and Procedure~\ref{alg:rcv_max}).
			\Comment{At the end of step one, the {\em root} contains the max (min) value.}\\
			\State \textbf{Step 2:} \textbf{Report} the element of the {\em root} and \textbf{reset} the flags for the searching ({\em in-memory searching}). During \textbf{reset}, a disabled link is switched to enable if it was disabled during Step 1.\\
			\State \textbf{Step 3:} The Phase~1 of {\em in-memory searching} is performed using Procedure~\ref{alg:send} and Procedure~\ref{alg:rcv}. After that, the $l_{m}$ flag of a node is disabled if its {\em match} flag is $1$ (the node contains the {\em search key}).
			\Comment{The memory of a node does not participate in the upcoming (next) {\em computing max (min)} procedure if its $l_{m}$ is disabled during this step.}\\
			\State \textbf{Step 4:} \textbf{Reset} the flags for the next {\em computing max (min)} procedure. During this, the disabled $l_{m}$ is not switched to enable. \\
			\State \textbf{Step 5:} Repeat Step 1 to Step 4 until the $l_{\mathscr{B}_j}$ of all the nodes are disabled.			
		\end{algorithmic}
	\end{scriptsize}
	\caption{\begin{scriptsize}In-memory sorting\end{scriptsize}}\label{alg:in-memory-sorting}
\end{algorithm}

\begin{table*}[ht]
	\centering
	\caption{Comparison of proposed {\em in-memory sorting} with the standard algorithms}	
	\begin{adjustbox}{width=1\textwidth,center}
		\begin{tabular}{|l|l|l|l|}
			\hline
			\textbf{Sorting algorithms}                                                                         & \textbf{\begin{tabular}[c]{@{}l@{}}Best case time \\ complexity\end{tabular}} & \textbf{\begin{tabular}[c]{@{}l@{}}Average case time \\ complexity\end{tabular}} & \textbf{\begin{tabular}[c]{@{}l@{}}Worst case time \\ complexity\end{tabular}} \\ \hline
			Insertion sort                                                                    & $\mathcal{O}(n)$                                                              & $\mathcal{O}(n^2)$                                                               & $\mathcal{O}(n^2)$                                                             \\ \hline
			Selection sort                                                                    & $\mathcal{O}(n^2)$                                                            & $\mathcal{O}(n^2)$                                                               & $\mathcal{O}(n^2)$                                                             \\ \hline
			Bubble sort                                                                       & $\mathcal{O}(n^2)$                                                            & $\mathcal{O}(n^2)$                                                               & $\mathcal{O}(n^2)$                                                             \\ \hline
			Merge sort                                                                        & $\mathcal{O}(n\log{n})$                                                       & $\mathcal{O}(n\log{n})$                                                          & $\mathcal{O}(n\log{n})$                                                        \\ \hline
			Heap sort                                                                         & $\mathcal{O}(n\log{n})$                                                       & $\mathcal{O}(n\log{n})$                                                          & $\mathcal{O}(n\log{n})$                                                        \\ \hline
			Quick sort                                                                        & $\mathcal{O}(n\log{n})$                                                       & $\mathcal{O}(n\log{n})$                                                          & $\mathcal{O}(n^2)$                                                             \\ \hline
			Radix sort                                                                        & $\mathcal{O}(nd)$                                                             & $\mathcal{O}(nd)$                                                                & $\mathcal{O}(nd)$                                                              \\ \hline
			\textbf{\begin{tabular}[c]{@{}l@{}}In-memory sorting\end{tabular}} & \textbf{$\mathcal{O}(\log{n})$}                                               & \textbf{$\mathcal{O}(n\log{n})$}                                                 & \textbf{$\mathcal{O}(n\log{n})$}                                               \\ \hline
		\end{tabular}
		
	\end{adjustbox}
	
	\label{Table:comp-sort}
\end{table*}


Table~\ref{Table:comp-sort} provides the comparison of the proposed {\em in-memory sorting} with the standard sorting algorithms on various parameters, where each column of the table indicates the best case, average case, and worst case time complexities of different sorting algorithms. The last row shows the time complexities of the proposed {\em in-memory sorting}. The best case scenario of the in-memory sorting is -- if all the elements in the input list are same.

\section{FPGA Implementation of IMC Platform}\label{hardware}
The {\em in-memory sorting} follows the {\em in-memory searching} and {\em computing max in-memory}. The in-memory computing (IMC) platform defined around the Cayley tree is realized with FPGA for evaluation. In this section, we present the realization of {\em in-memory searching} in FPGA for illustration and evaluation. Two types of designs are realized. One is around a new memory architecture (NMA) defined for the proposed IMC platform, and the other one is based on the existing memory architecture (EMA). The following environment is considered for the IMC platform. 

\subsubsection*{Hardware}
Intel$\textsuperscript{\textregistered}$ Core$\textsuperscript{\texttrademark}$ $i7+8700$~ $\times64$-based Processor ($12$M Cache) with Intel$\textsuperscript{\textregistered}$ Optane$\textsuperscript{\texttrademark}$ Memory; and $8$ GB ($7.78$ GB usable) installed RAM.

\subsubsection*{Software}
$64-$bit OS; the system edition ``Windows $11$ Pro'' of version ``$22H2$''; and Vivavo ML $2023.1$ - standard. The simulation is done using Xilinx Vivado.


		%

\subsection{IMC around NMA}\label{DNMA}
%

%
\begin{figure}[h! tbp]
		\begin{center}
		\scalebox{1}{
			\begin{tabular}{c}
				\begin{adjustbox}{addcode={\begin{minipage}{\width}}{\end{minipage}},rotate=90}
					
					\includegraphics[width=14cm,height=12cm]{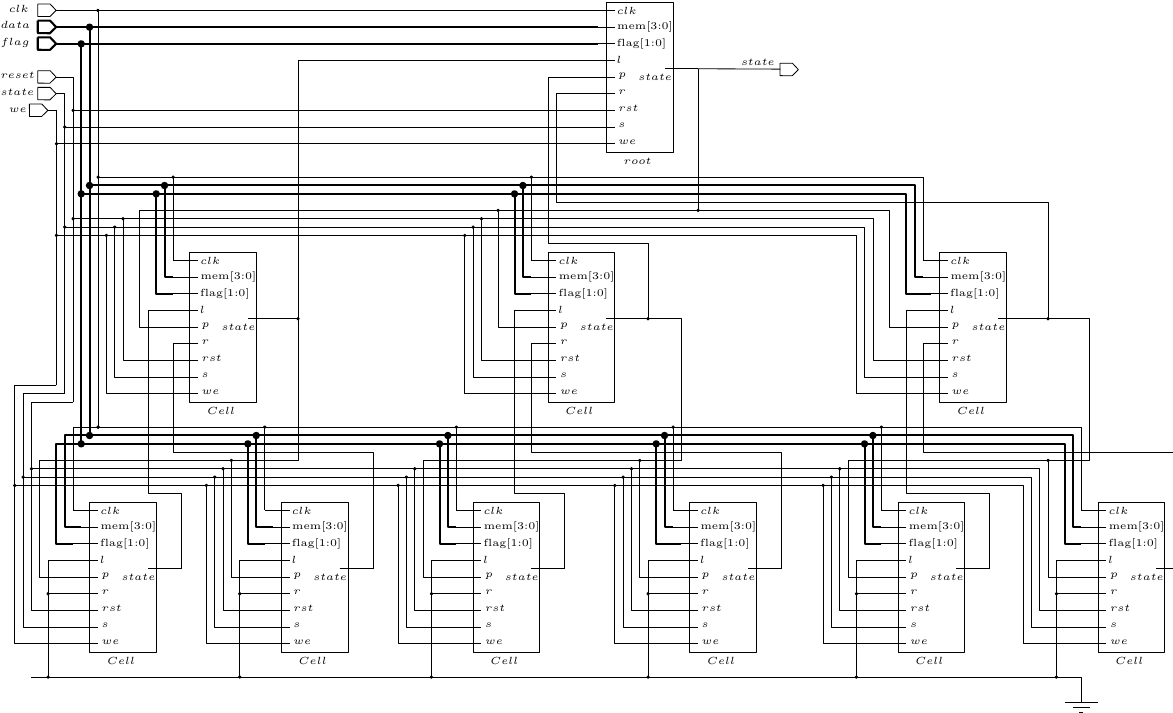}
				\end{adjustbox}\\
				a.\\
				\begin{adjustbox}{addcode={\begin{minipage}{\width}}{\end{minipage}},rotate=90}
					\includegraphics[width=4cm,height=12cm]{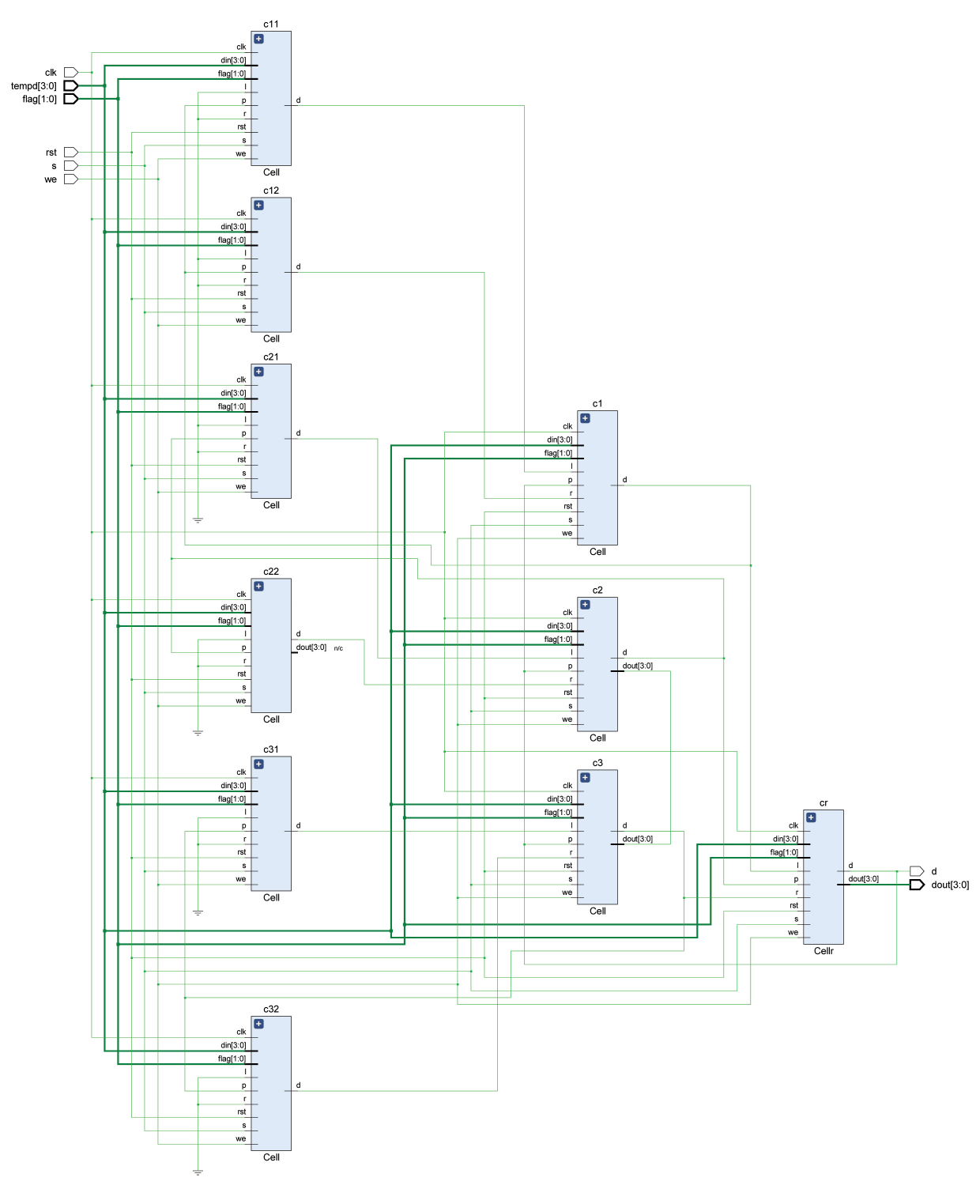}
				\end{adjustbox}
				\\
				b.\\
			\end{tabular}
	}
		\caption{\scriptsize RTL schematic of NMA platform; (a) RTL schematic, (b) Simulation in Vivado}
		\label{Intree}
			\end{center}
\end{figure}
\begin{figure}[h! tbp]
\begin{center}
	\renewcommand{\arraystretch}{-2}
	\scalebox{1}{
		\begin{tabular}{cc}
			\begin{adjustbox}{addcode={\begin{minipage}{\width}}{\end{minipage}},rotate=90}
				
				\includegraphics[width=1.2\textwidth]{ 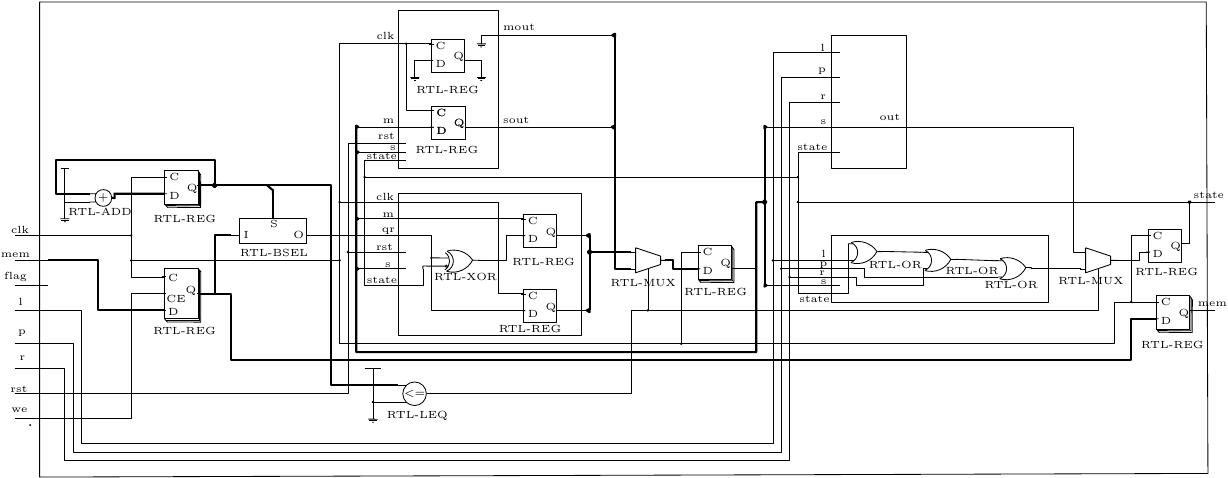}
			\end{adjustbox}&
			\begin{adjustbox}{addcode={\begin{minipage}{\width}}{\end{minipage}},rotate=90}
				\includegraphics[width=1.2\textwidth]{ 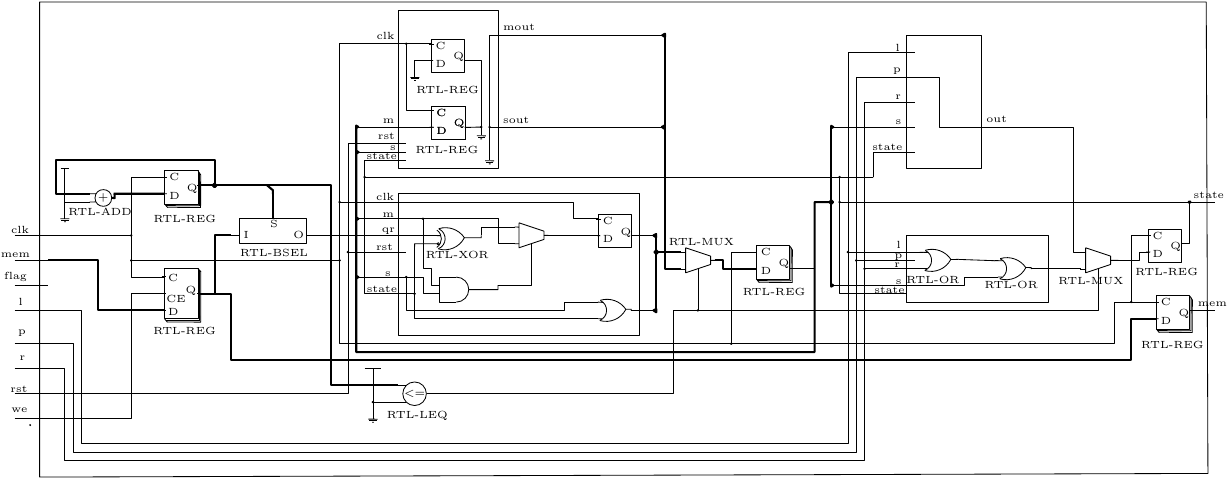}
			\end{adjustbox}
			\\
			a.&b.\\
	\end{tabular}}
	\caption{RTL schematic of NMA nodes; (a) {\em root} node, (b) Intermediate node. Ports $l,~p,$ and $r$ represent the states of left child, middle child/parent and right child. The $clk$ and $rst$ represent the clock and reset}
	\label{Inroot}
\end{center}
\end{figure}

The RTL schematic of NMA based design is shown in Figure~\ref{Intree} (Figure~\ref{Intree}a. is the RTL schematic and Figure~\ref{Intree}b. is the screenshot of the same RTL schematic simulation in Vivado), where the top cell is the {\em root} node and the other $9$ are the intermediate nodes and leaves of the Cayley tree. Initially, the $ data$ port is used to store an element to the memory word of a node. An input port $state$ is used to set/reset the {\em state} flag. The `$flag$' is used to set/reset the {\em start} and {\em match} flags. The input port, `$we$' is a `write enable' port. It is activated to store the element in memory word. The detailed RTL schematic of the {\em root} and intermediate node are shown in Figure~\ref{Inroot}. 

The behavioral analysis of the IMC platform for different word size is done to find the utilization of BUFG (global buffer), I/O, FF, LUT (look up table) and fully routed nets. For experimentation, the word size, ranging from $4$ to $64$ bits, is considered. 
Table~\ref{Table:util} reports the estimated utilization of NMA based realization. The second to tenth columns of the table note the results for different word size. The second row shows the utilization of BUFG. The third row indicates the I/O requirements for different word size as well as the percentage of total number of I/Os. Similarly, the fourth, fifth and sixth rows indicate the required number of Flip-Flops, LUTs and the fully routed nets respectively.
	
Table~\ref{Table:power} reports the power summary of NMA based platform. The second to tenth columns of the table note the results for different word size. The rows describe different parameters related to the power. The second row shows the Total on-chip power (W) that includes static power (leakage power) and dynamic power. The third row indicates the static power consumption of the programmable logic (PL) section. The fourth row indicates the power consumed by the I/O blocks that includes dynamic I/O power and static I/O power. The I/O interface determines how the model communicates with the other components in the system such as, sensors/ processors/ memory/ other peripherals. The key performance matrices of I/O are the data rate, signal integrity, timing and setup/hold margins, drive strength and slew rate control. The fifth row shows the power consumed by the programmable logic fabric that includes dynamic logic power (switching power, clock power etc.) and the static logic power (leakage power). The sixth row reports the  power consumed by the internal and external signals. The junction temperature (T$_j$) which refers to the temperature at the actual silicon die -that is, the `junction' where the transistors are located inside the FPGA, is shown in the seventh row. This parameter is very crucial and monitored during the development of NMA based platform. If the maximum allowable junction temperature exceeds, then it may cause malfunction. The eighth row shows the thermal margin which is the difference between the maximum allowable junction temperature (T$_{j(max)}$) and the actual operating junction temperature (T$_j$). Here, a positive ($+ve$) thermal margin implies that the operating temperature of the FPGA is below the maximum allowable temperature whereas, a negative ($-ve$) thermal margin occurs when the operating temperature exceeds the maximum allowable temperature. The last row of Table~\ref{Table:power} points to the effective instruction-level just-in-time acceleration ($\mathscr{I}JA$) which involves in dynamically optimizing the hardware execution of task at a very fine-grain level.
	
	
%

\begin{table*}[ht]
\centering
\caption{Estimated utilization in NMA based design}	
\begin{adjustbox}{width=1\textwidth,center}
	\begin{tabular}{|l|l|l|l|l|l|l|l|l|l|}
		\hline
		\textbf{\makecell{Estimated \\utilization}}  & \textbf{$4$-bit}                                         & \textbf{$8$-bit} & \textbf{$16$-bit} & \textbf{$24$-bit} & \textbf{$32$-bit} & \textbf{$40$-bit} & \textbf{$48$-bit} & \textbf{$56$-bit} & \textbf{$64$-bit} \\ \hline
		\textbf{\makecell{BUFG}}& \makecell{$1$\\$(3.13\%)$}  & \makecell{$1$\\$(3.13\%)$}  & \makecell{$1$\\$(3.13\%)$}   & \makecell{$1$\\$(3.13\%)$}  & \makecell{$1$\\$(3.13\%)$}   & \makecell{$1$\\$(3.13\%)$}   & \makecell{$1$\\$(3.13\%)$}   & \makecell{$1$\\$(3.13\%)$}    & \makecell{$1$\\$(3.13\%)$}    \\ \hline
		\textbf{\makecell{I/O} }                & \makecell{$11$\\$(3.67\%)$} & \makecell{$21$\\$(7.0\%)$}  & \makecell{$35$\\$(11.67\%)$} & \makecell{$51$\\$(17\%)$}   & \makecell{$67$\\$(22.33\%)$} & \makecell{$83$\\$(27.67\%)$} & \makecell{$99$\\$(33\%)$}    & \makecell{$115$\\$(38.33\%)$} & \makecell{$131$\\$(43.67\%)$}  \\ \hline
		\textbf{\makecell{FF}  }                  & \makecell{$55$\\$(0.07\%)$} & \makecell{$65$\\$(0.08\%)$} & \makecell{$79$\\$(0.10\%)$}  & \makecell{$95$\\$(0.12\%)$} & \makecell{$111$\\$(0.14\%)$} & \makecell{$127$\\$(0.15\%)$} & \makecell{$143$\\$(0.17\%)$} & \makecell{$159$\\$(0.19\%)$}  & \makecell{$175$\\$(0.21\%)$}  \\ \hline
		\textbf{\makecell{LUT} }                  & \makecell{$29$\\$(0.07\%)$} & \makecell{$31$\\$(0.08\%)$} & \makecell{$32$\\$(0.08\%)$}  & \makecell{$35$\\$(0.09\%)$} & \makecell{$37$\\$(0.09\%)$}  & \makecell{$39$\\$(0.10\%)$}  & \makecell{$41$\\$(0.10\%)$}  & \makecell{$43$\\$(0.01\%)$}   & \makecell{$45$\\$(0.11\%)$}   \\ \hline
		
		\textbf{\makecell{Fully\\ routed\\ nets}}& $79$& $91$& $111$&$138$& $163$& $186$ & $211$& $235$& $259$ \\ \hline
	\end{tabular}
\end{adjustbox}

\label{Table:util}
\end{table*}


\begin{table*}[h]
\centering
\caption{Power summary report for NMA based design}	
\begin{adjustbox}{width=1\textwidth,center}	
	\begin{tabular}{|l|l|l|l|l|l|l|l|l|l|}
		\hline
		\textbf{\makecell{Power \\summary}}    & \textbf{$4$-bit}                                         & \textbf{$8$-bit} & \textbf{$16$-bit} & \textbf{$24$-bit} & \textbf{$32$-bit} & \textbf{$40$-bit} & \textbf{$48$-bit} & \textbf{$56$-bit} & \textbf{$64$-bit} \\ \hline
		
		\textbf{\makecell{Total\\ on-chip\\ power (W)}}      & $4.661$& $8.548$ & $13.312$ & $19.347$ & $25.112$ & $31.052$  & $37.006$  & $43.595$   & $49.801$\\ \hline
		\textbf{\makecell{PL \\static (W)}}                & 0.093                               & 0.105                               & 0.125                                & 0.165                                & 0.226                                & 0.315                                & 0.444                                & 0.623                                & 0.867                                \\ \hline
		\textbf{\makecell{I/O (W)}}                      & 4.214                               & 7.989                               & 12.607                               & 18.417                               & 24.219                               & 30.008                               & 35.830                               & 41.636                               & 47.446                               \\ \hline
		\textbf{\makecell{Logic (W)}}                    & 0.122                               & 0.141                               & 0.170                                & 0.212                                & 0.213                                & 0.231                                & 0.247                                & 0.272                                & 0.285                                \\ \hline
		\textbf{\makecell{Signals (W)}}                  & 0.232                               & 0.314                               & 0.409                                & 0.552                                & 0.681                                & 0.813                                & 0.929                                & 1.063                                & 1.202                                \\ \hline
		\textbf{\makecell{Junction \\temperature\\ ($^oC$)}} & 33.8                                & 41.1                                & 50.1                                 & 61.4                                 & 72.7                                 & 84.1                                 & 95.5                                 & 107.1                                & 118.8                                \\ \hline
		\textbf{\makecell{Thermal \\margin\\ ($^oC$~($W$))}}       & 51.2 (27)                           & 43.9 (23.1)                         & 34.9 (18.4)                          & 23.6 (12.4)                          & 12.3 (6.5)                           & 0.9 (0.5)                            & -10.5 (-5.4)                         & -22.1 (-11.4)                        & -33.8 (-17.4)                        \\ \hline
		\textbf{\makecell{Effective \\$\mathscr{I}JA$ ($^oC/W$)}}      & 1.9                                 & 1.9                                 & 1.9                                  & 1.9                                  & 1.9                                  & 1.9                                  & 1.9                                  & 1.9                                  & 1.9                                  \\ \hline
		
	\end{tabular}
\end{adjustbox}

\label{Table:power}
\end{table*}

%
%



		%

%

It can be observed from Table~\ref{Table:power} that for the word size greater than $40$-bit, the junction temperature exceeds due to the use of large number of bits in a single node.
To overcome this, the storage can be kept out of a node. The evaluation of such a design follows next. 
%
%

\subsection{IMC around EMA}

This design is based on the existing RAM architecture.
\begin{figure}[h! tbp]
	\begin{center}
	\scalebox{1}{
			\begin{tabular}{c}
			\begin{adjustbox}{addcode={\begin{minipage}{\width}}{\end{minipage}},rotate=90}
				\includegraphics[width=14cm,height=12cm]{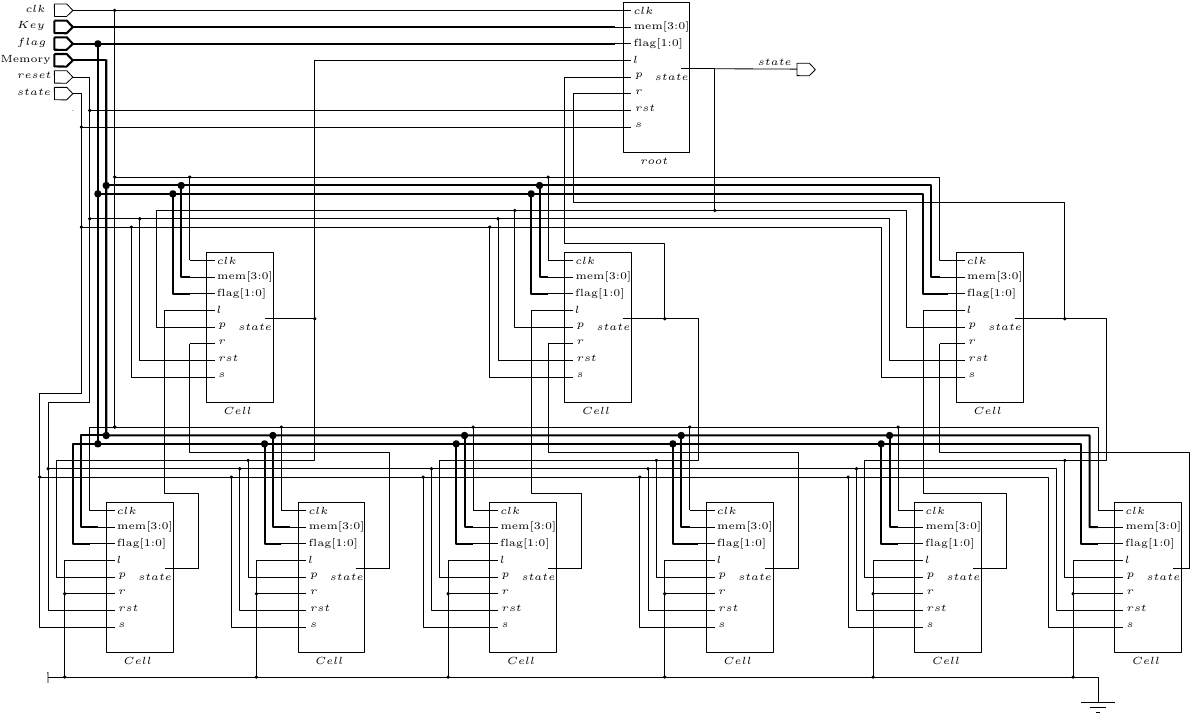}
			\end{adjustbox}\\
			a.\\
			\begin{adjustbox}{addcode={\begin{minipage}{\width}}{\end{minipage}},rotate=90}
				\includegraphics[width=4cm,height=12cm]{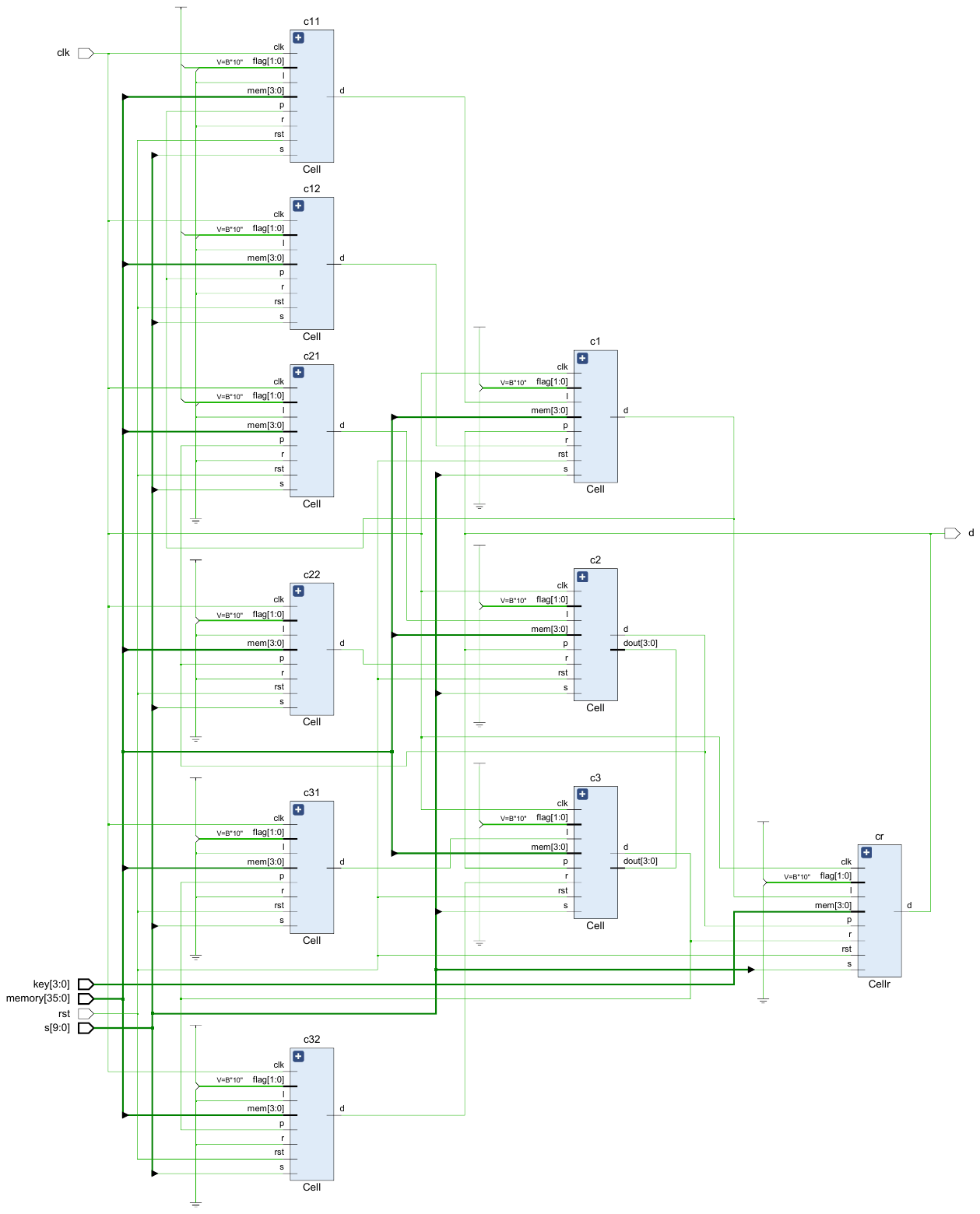}
			\end{adjustbox}
			\\
			b.\\
		\end{tabular}
}
	\caption{Schematic of the EMA platform; (a) RTL schematic, (b) Simulation in Vivado}
	\label{H1}
		\end{center}
	\end{figure}
	The RTL schematic of EMA platform, in FPGA, with {\em root} and other $9$ nodes is shown in Figure~\ref{H1} (Figure~\ref{H1}a. is the RTL schematic and Figure~\ref{H1}b. is the screenshot of the same RTL schematic simulation in Vivado). The structures of an intermediate node and the {\em root} node are shown in Figure~\ref{H2}.
	%
		%
\begin{figure*}[h! tbp]
\begin{center}
	\renewcommand{\arraystretch}{-5}
	\scalebox{1}{
		\begin{tabular}{cc}
			\begin{adjustbox}{addcode={\begin{minipage}{\width}}{\end{minipage}},rotate=90}
				\includegraphics[width=1.2\textwidth] {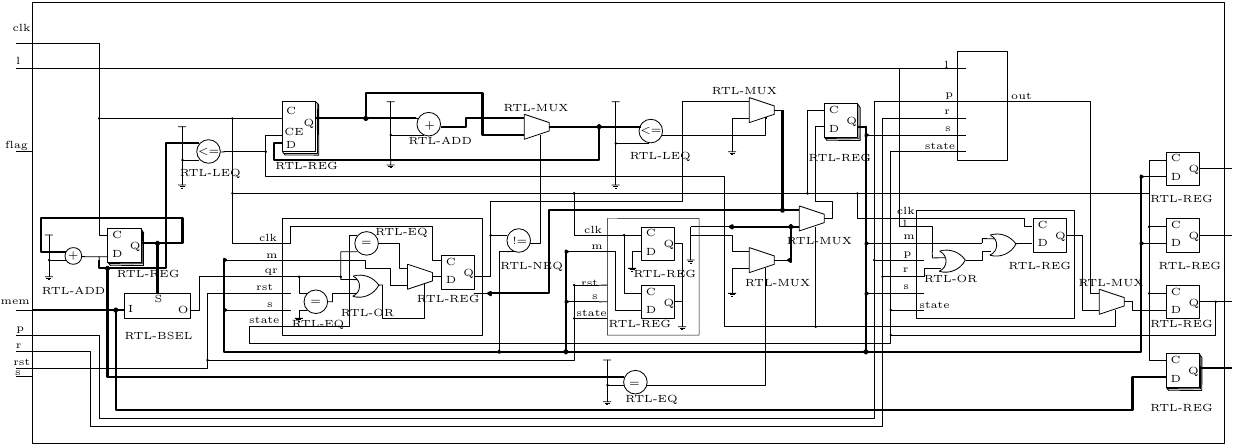}
			\end{adjustbox}&
			\begin{adjustbox}{addcode={\begin{minipage}{\width}}{\end{minipage}},rotate=90}
				\includegraphics[width=1.2\textwidth]{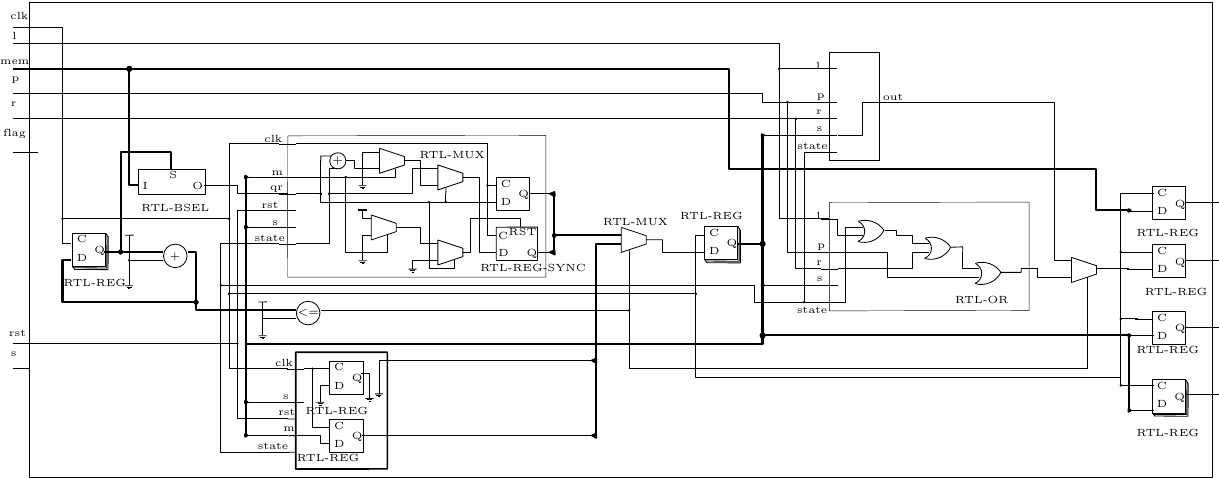}
			\end{adjustbox}
			\\
			a.&b.\\
	\end{tabular}}
	\caption{RTL schematic of EMA nodes; (a) Intermediate node, (b) {\em root} node. Ports $l,~p,$ and $r$ represent the states of left child, parent/middle node and right child. The $clk$ and $rst$ represent the clock and reset}
	
	\label{H2}
\end{center}
\end{figure*}
Table~\ref{Table:util2} and Table~\ref{Table:power2} report the post synthesis estimated utilization and power summary for $4$, $8$, and $16$-bit words, where each column ($2^{nd}$ to $4^{th}$ column) of the tables notes the results for different word size. The parameters chosen for experimentation in EMA based design are same as that of Section~\ref{DNMA}. As this design involves a huge number of nets ($4^{th}$ column of Table~\ref{Table:util2}), the simulation cannot be concluded for a word size greater than the $16$-bit with the Vivado's downloadable version.

%



\begin{table*}[ht]
	\centering
	\caption{Estimated utilization in EMA based design}	
	\begin{adjustbox}{width=0.6\textwidth,center}
		\begin{tabular}{|l|c|c|c|}
			\hline
			\textbf{\makecell{Estimated utilization}} & \multicolumn{1}{l|}{\textbf{$4$-bit}}                  & \multicolumn{1}{l|}{\textbf{$8$-bit}}                  & \multicolumn{1}{l|}{\textbf{$16$-bit}}                 \\ \hline
			BUFG                                                   & \begin{tabular}[c]{@{}c@{}}1\\ (3.13\%)\end{tabular}   & \begin{tabular}[c]{@{}c@{}}1\\ (3.13\%)\end{tabular}   & \begin{tabular}[c]{@{}c@{}}1\\ (3.13\%)\end{tabular}   \\ \hline
			I/O                                                    & \begin{tabular}[c]{@{}c@{}}42 \\ (14\%)\end{tabular}   & \begin{tabular}[c]{@{}c@{}}82\\ (27.3\%)\end{tabular}  & \begin{tabular}[c]{@{}c@{}}162\\ (54\%)\end{tabular}   \\ \hline
			FF                                                     & \begin{tabular}[c]{@{}c@{}}359\\ (0.45\%)\end{tabular} & \begin{tabular}[c]{@{}c@{}}359\\ (0.45\%)\end{tabular} & \begin{tabular}[c]{@{}c@{}}647\\ (0.82\%)\end{tabular} \\ \hline
			LUT                                                    & \begin{tabular}[c]{@{}c@{}}217\\ (0.53\%)\end{tabular} & \begin{tabular}[c]{@{}c@{}}227\\ (0.55\%)\end{tabular} & \begin{tabular}[c]{@{}c@{}}404\\ (0.99\%)\end{tabular} \\ \hline
			Fully routed nets                                      & 913                                                    & 962                                                    & 1788                                                   \\ \hline
		\end{tabular}
		
	\end{adjustbox}
	
	\label{Table:util2}
\end{table*}

\begin{table*}[ht]
\centering
\caption{Power summary for EMA based design}
\begin{adjustbox}{width=0.6\textwidth,center}		
	\begin{tabular}{|l|l|l|l|}
		\hline
		\textbf{\makecell{Power summary}}    & \textbf{$4$-bit}                                         & \textbf{$8$-bit} & \textbf{$16$-bit}\\ \hline
		
		Total on-chip power (W)      & 4.042                                & 5.076                               & 7.489                                \\ \hline
		PL static (W)               & 0.091                               & 0.094                               & 0.101                                \\ \hline
		I/O (W)                      & 1.603                               & 2.142                               & 1.655                                \\ \hline
		Logic (W)                    & 0.994                               & 1.138                               & 2.223                                \\ \hline
		Signals (W)                  & 1.354                               & 1.703                               & 3.509                                \\ \hline
		Junction temperature ($^oC$) & 32.6                                & 34.6                                & 39.1                                 \\ \hline
		Thermal margin ($^oC$)       & 52.4                                & 50.4                                & 45.9                                 \\ \hline
		Effective $\mathscr{I}JA$ ($^oC/W$)      & 1.9                                 & 1.9                                 & 1.9                                  \\ \hline
		
	\end{tabular}
\end{adjustbox}

\label{Table:power2}

\end{table*}

\subsection{Comparison of NMA and EMA platforms}

\begin{figure}[hbt!]
	\begin{center}
		\scalebox{1}{
			\begin{tabular}{cc}
				\includegraphics[width=0.55\textwidth]{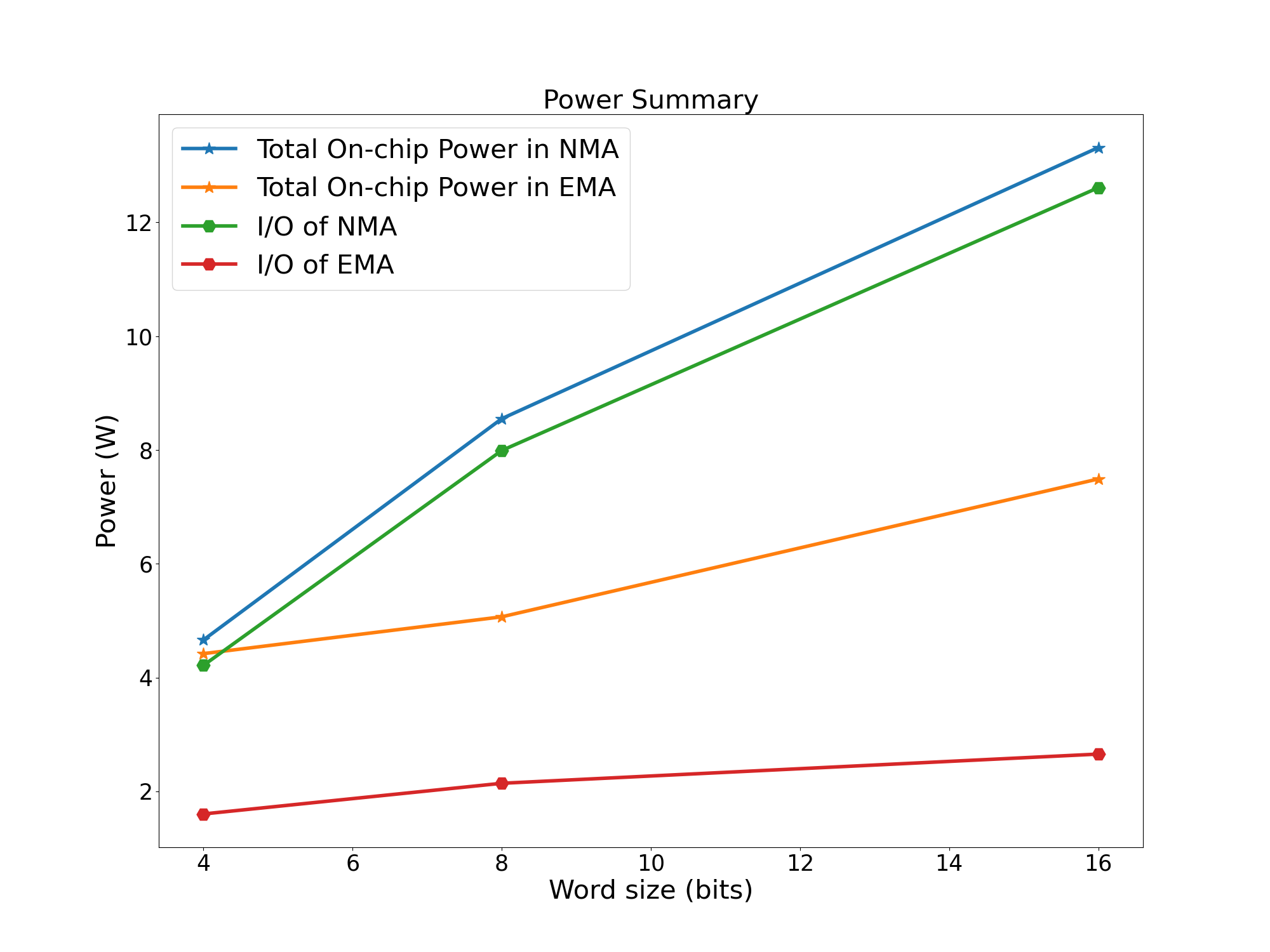}&\hspace{-3 em}
				\includegraphics[width=0.55\textwidth]{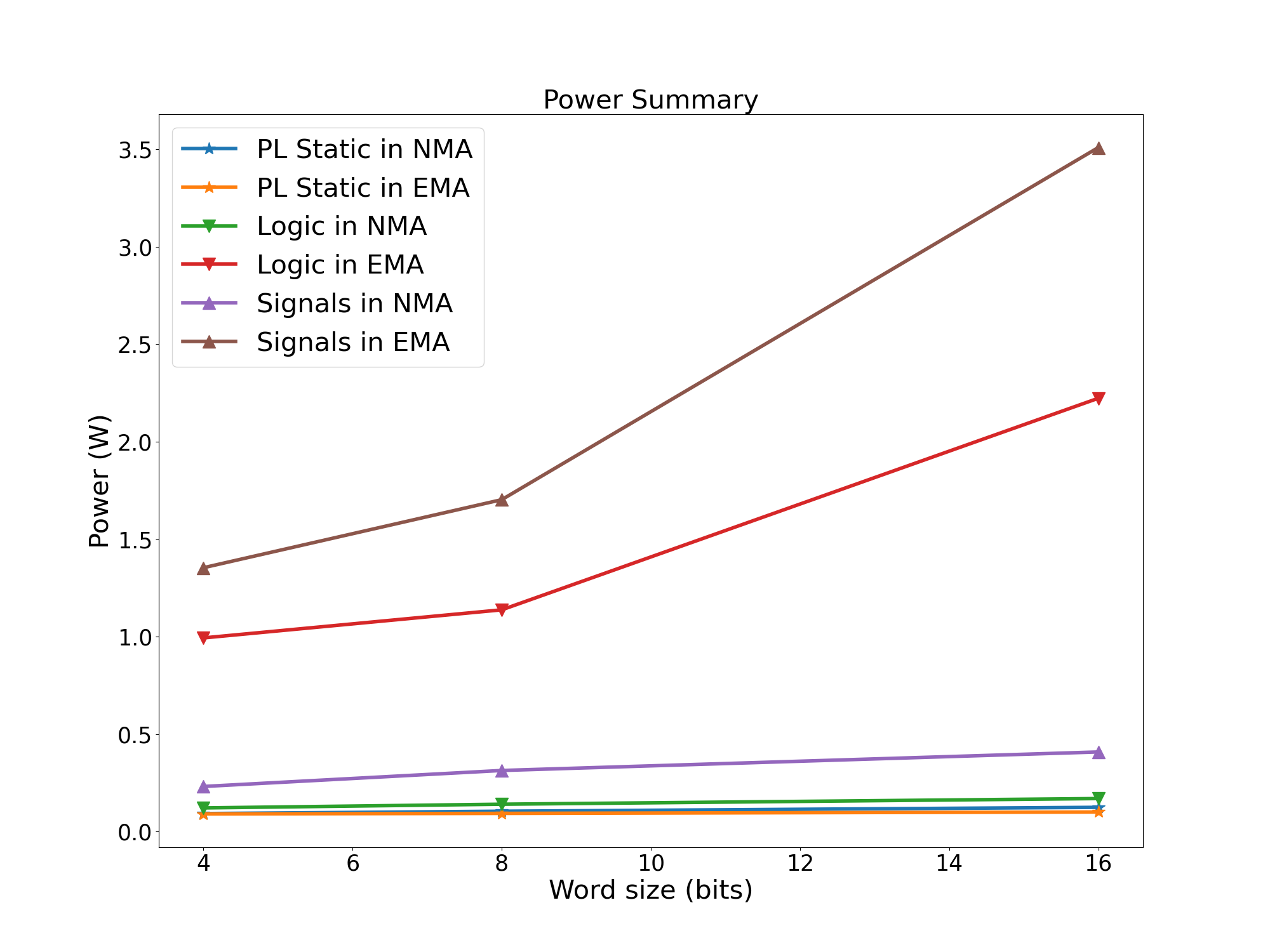}\\
				(a)&(b)\\
		\end{tabular}}
		
		\scalebox{1}{
			\begin{tabular}{ccc}
				&\includegraphics[width=0.95\textwidth]{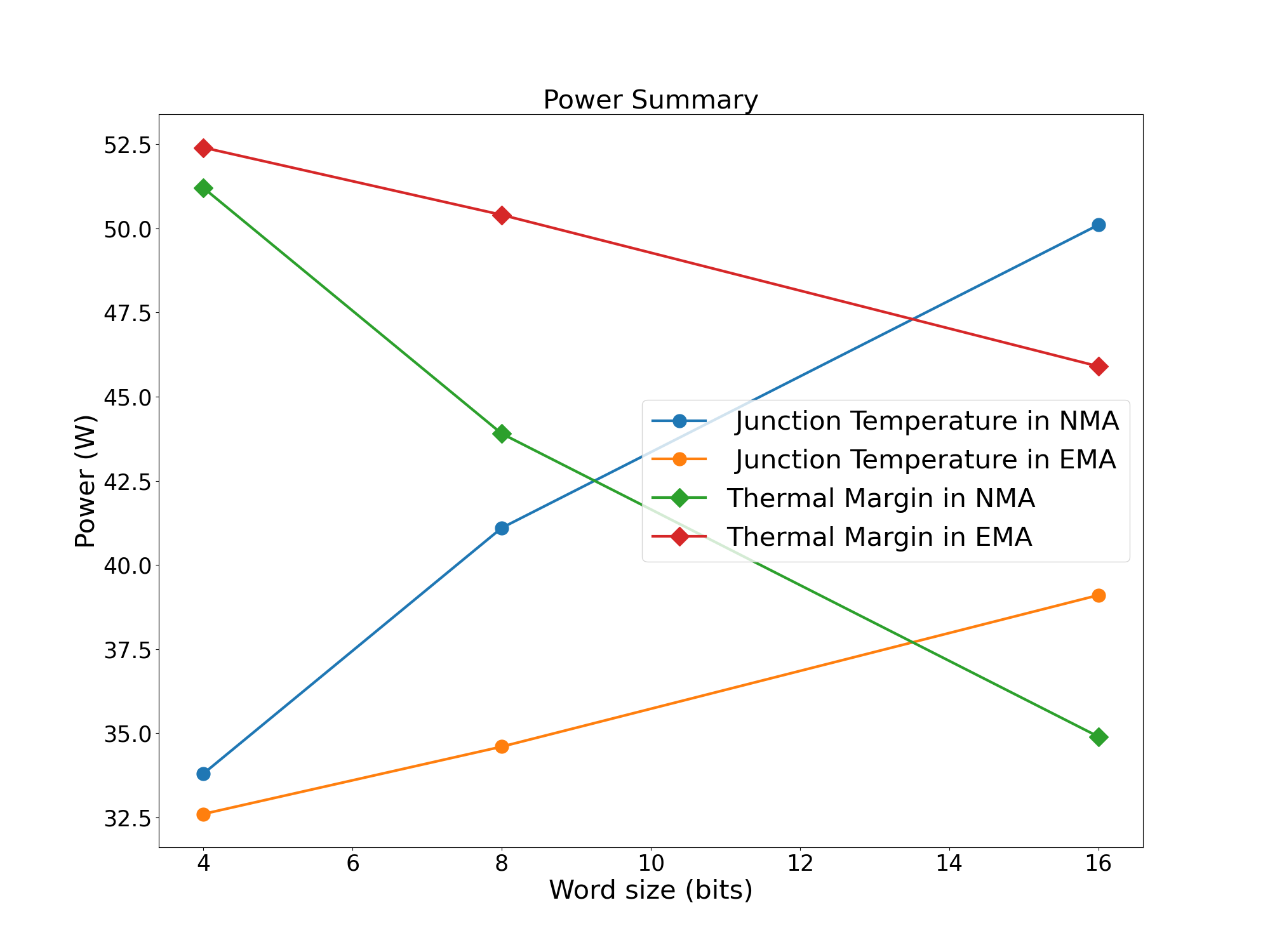}&\\
				&(c)&\\
		\end{tabular}}
		\caption{Comparison of power summary of NMA and EMA}
		\label{comgraph1}
	\end{center}
\end{figure}


The performances of both the IMC platforms (NMA and EMA based) are not limited to the scale of memory. 
In both the NMA and EMA schemes --
\begin{itemize}
\item A larger word size indicates higher activity, more transitions and logic components etc. These are the causes of higher power consumption. All these add additional dynamic power/junction leakage (this increases static power)/static power in programmable logic section and in I/O components/power consumption by the signals within a device.	However, the effective $ \mathscr{I}JA$ ($^oC/W$) is constant for both the schemes (Table~\ref{Table:power} and Table~\ref{Table:power2}).

\item The BUFG in the NMA and EMA platform are $1$ (Table~\ref{Table:util} and Table~\ref{Table:util2}).

\item The larger word size increases the requirement of FFs, LUTs and fully routed nets leading to more power (Table~\ref{Table:util} and Table~\ref{Table:util2}).
\end{itemize}    	

The comparison of power summary of the IMC platforms around NMA and EMA is shown in  Figure~\ref{comgraph1}. Figure~\ref{comgraph1}(a) shows the comparison of Total on-chip power and I/O power in NMA and EMA models for different word size.  Figure~\ref{comgraph1}(b) shows the details of PL static power required for Logic and Signals in NMA and EMA models. The comparison of Junction Temperature and Thermal Margin of NMA and EMA models are shown in Figure~\ref{comgraph1}(c).  
It can be observed from the figures that
\begin{itemize}
	\item The different power consumption in NMA based design -that is, PL static, I/O and junction temperature are more than that of the EMA based design. However, the power consumption due to logic, signals and thermal margin are less in NMA compared to EMA.
	
	\item The number of clocking and data toggling in NMA based design is also greater than that of the EMA based design. It increases the total on-chip power of NMA based design.
	
\end{itemize}

The comparison of estimated utilization in the designs is shown in Figure~\ref{comgraph2}. 
It shows that the EMA based design requires more resources -that is, I/Os, FFs, LUTs and fully routed nets than that of NMA based design.

\begin{figure*}[ht]
\begin{center}
	\scalebox{1}{
		\begin{tabular}{c}
			\includegraphics[width=1.01\textwidth]{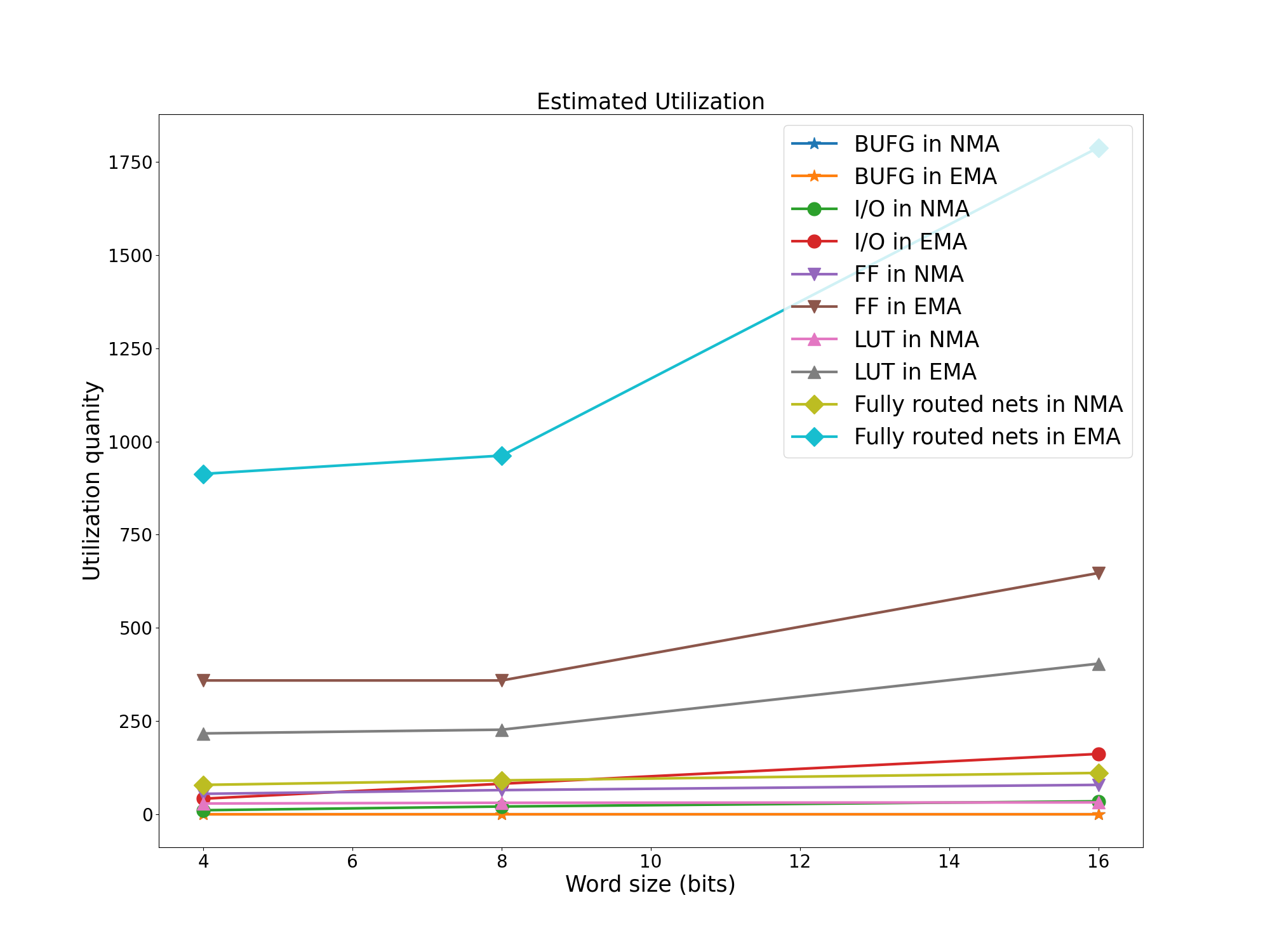}\\
		\end{tabular}}
		\caption{Comparison of NMA and EMA based IMC on estimated utilization}
		\label{comgraph2}
	\end{center}
\end{figure*} 

\subsection{Comparison of proposed IMC with state-of-the art designs}\label{SOA_comp}

In this work, we propose an IMC architecture and have shown how sorting, searching, and computation of max (min) can be done in-memory. 
The Content Addressable Memory (CAM) can be the best choice for {\em in-memory searching}. It completes the comparison of a key with all the elements of a list in one unit of time but at the cost of a dedicated comparison logic in memory. 

The proposed NMA based design is compared with the design (searching) reported in~\cite{ullah2012fpga}. It is a $512\times36$ SRAM based Ternary Content Addressable Memory (SR-TCAM). The SR-TCAM achieves comparable searching operation in two clock cycles and targets large capacity TCAMs. The four design options Design 1, Design 2, Design 1 with PE and Design 2 with PE are reported in~\cite{ullah2012fpga}. The NMA based design is compared with these four designs in terms of the required FF and LUTs. Figure~\ref{H11} illustrates the comparison results, where the dimensions of SR-TCAM and NMA are considered to be the same --that is, $512\times36$. Different color bars indicate different models. It can be observed that the NMA based design requires more FFs than the other models, because the additional three flags are taken into account by each node in NMA based design. However, the requirement of LUTs in NMA is very less compared to the other models.

Since the effective realization of {\em in-memory sorting} was not available in the literature and, therefore, cannot be compared with the proposed {\em in-memory sorting} scheme.


\begin{figure}[ht]
	\begin{center}
		\scalebox{1}{
			\begin{tabular}{c}
				\includegraphics[width=0.7\columnwidth]{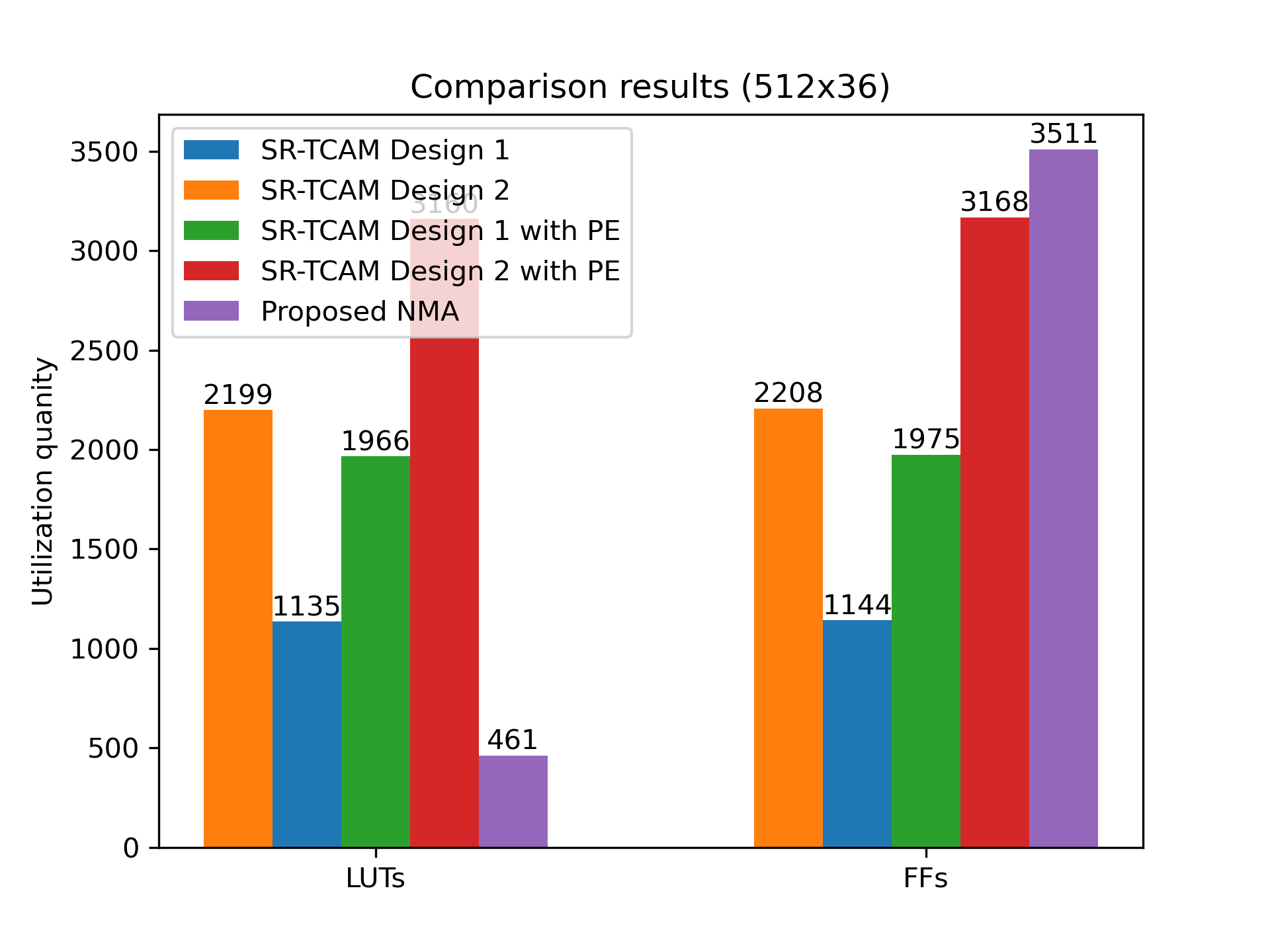}
		\end{tabular}}
		\caption{Comparison of SR-TCAM~\cite{ullah2012fpga} and NMA based design of size $512\times36$}
		\label{H11}
	\end{center}
	
\end{figure}

\section{Conclusion}\label{conclusion}
This work proposes an effective solution for computational problems -that is, {\em in-memory searching}, {\em computing max (min) in-memory} and {\em in-memory sorting}. The proposed {\em in-memory sorting} follows the {\em in-memory searching} and {\em computing max (min)}. The in-memory computing (IMC) platform has been developed around the Cayley tree. Each node of the tree is associated with a word of memory, and is driven by a function. The IMC platform avoids data movements between the memory and CPU and, therefore, very low bandwidth is required for communication. The comparison logic is the only overhead in memory.

For the IMC, two type of designs have been proposed, one is around a new memory architecture (NMA) and the other one is based on the existing memory architecture (EMA). The IMC platform has shown massive parallelism in computation that enables {\em sorting} in $\mathcal{O}(n\log{n})$, {\em computing max (min) in-memory} in $\mathcal{O}(\log{n})$ and  {\em searching} in $\mathcal{O}(\log{n})$ time, even for the unsorted input list. The space complexity of these scheme is $\mathcal{O}(n)$. The efficacy of the proposed IMC platform is established while compared with the State-of-the art solutions.

\bibliographystyle{els}
\bibliography{Cayleysearch}

\end{document}